\documentclass[11pt]{amsart}

\usepackage{lipsum}
\usepackage{amsfonts}
\usepackage{graphicx}

\usepackage{amssymb,amsmath}

\usepackage{dsfont}
\usepackage[colorlinks, citecolor=black]{hyperref}
\usepackage{color}

\usepackage{mathrsfs}

\usepackage{subfigure}

\usepackage{caption}
\usepackage{bm}
\usepackage{algorithm}
\usepackage{algorithmicx}
\usepackage{algpseudocode}
\usepackage{tikz}
\usetikzlibrary{patterns} 

\newcommand{\xkh}[1]{\left(#1\right)}

\newcommand{\dkh}[1]{\left\{#1\right\}}
\newcommand{\zkh}[1]{\left[#1\right]}

\newcommand{\nj}[1]{\langle {#1} \rangle}

\newcommand{\norm}[1]{\left\|{#1}\right\|_2}
\newcommand{\normf}[1]{\left\|{#1}\right\|_F}

\newcommand{\norms}[1]{\left\|{#1}\right\|}

\newcommand{\ucnorm}[1]{\|{#1}\|_2}

\newcommand{\ucnorms}[1]{\|{#1}\|}
\newcommand{\ucabs}[1]{\lvert#1\rvert}

\newcommand{\abs}[1]{\left\lvert#1\right\rvert}
\newcommand{\aabs}[1]{\lvert#1\rvert}

\newcommand{\argmin}[1]{\mathop{\rm argmin}\limits_{#1}}

\newcommand{\minm}[1]{\mathop{\rm min}\limits_{#1}}
\newcommand{\maxm}[1]{\mathop{\rm max}\limits_{#1}}

\newcommand{\sumim}{\sum_{i=1}^m}
\newcommand{\sumjm}{\sum_{j=1}^m}
\newcommand{\sumkm}{\sum_{k=1}^m}

\newcommand{\disth}[1]{\mathrm {dist}_\H\left(#1\right)}

\newcommand{\dd}{{\mathrm d}}
\newcommand{\e}{{\mathrm e}}

\renewcommand{\H}{{\mathbb H}}
\newcommand{\E}{{\mathbb E}}

\newcommand{\R}{{\mathbb R}}
\newcommand{\SSS}{{\mathbb S}}

\newcommand{\X}{{\mathcal X}}
\newcommand{\C}{{\mathbb C}}

\newcommand{\gk}{g_k}
\newcommand{\rk}{r_k}
\newcommand{\mem}{{{\bm E}_m}}
\newcommand{\hmd}{{\mathbb H}^{m\times d}}

\newcommand{\alphak}{\alpha_k}
\newcommand{\alphakk}{\alpha_{k+1}}
\newcommand{\phik}{\varphi_k}
\newcommand{\phikk}{\varphi_{k+1}}
\newcommand{\thetak}{\theta_k}

\newcommand{\zeropi}{\zkh{0,\pi}}

\newcommand{\betalp}{\beta^{\ell_p}}

\newcommand{\betaone}{\beta^{\ell_1}}
\newcommand{\betaonebar}{\bar{\beta}^{\ell_1}}

\newcommand{\betatwo}{\beta^{\ell_2}}
\newcommand{\betatwobar}{\bar{\beta}^{\ell_2}}

\newcommand{\real}[1]{\mathrm{Re}\left(#1\right)}
\newcommand{\image}[1]{\mathrm{Im}\left(#1\right)}
\newcommand{\ucreal}[1]{\mathrm{Re}(#1)}

\newcommand{\hh}{\ast}

\newcommand{\T}{\top}

\newcommand{\va}{{\bm a}}
\newcommand{\vb}{{\bm b}}

\newcommand{\vu}{{\bm u}}
\newcommand{\vv}{{\bm v}}
\newcommand{\vx}{{\bm x}}
\newcommand{\vy}{{ \bm{ y}}}

\newcommand{\vxt}{{\bm x}^t}

\newcommand{\ma}{{\bm A}}
\newcommand{\mb}{{\bm B}}

\newcommand{\mi}{{\bm I}}

\newcommand{\mpm}{{\bm P}}

\newcommand{\ve}{{\bm e}}

\renewcommand{\i}{{\rm i}}

\renewcommand{\omega}{\eta}

\newcommand{\sqrttwo}{\sqrt{2}}

\newcommand{\pitwo}{\frac{\pi}{2}}
\newcommand{\pim}{\frac{\pi}{m}}

\newcommand{\kpim}{\frac{k\pi}{m}}
\newcommand{\pimtwo}{\frac{\pi}{2m}}
\newcommand{\pimtwothree}{\frac{3\pi}{2m}}

\newcommand{\pimfourthree}{\frac{3\pi}{4m}}

\newcommand{\tjk}{\tau_j^k}

\newtheorem{definition}{Definition}[section]
\newtheorem{corollary}[definition]{Corollary}

\newtheorem{theorem}[definition]{Theorem}
\newtheorem{lemma}[definition]{Lemma}

\newtheorem{remark}[definition]{Remark}

\date{}

\begin{document}

\author{Haiyang Peng}
\address{School of Mathematical Sciences, Beihang University, Beijing, 100191, China }
\email{haiyangpeng@buaa.edu.cn}

\author{Deren Han}
\address{School of Mathematical Sciences, Beihang University, Beijing, 100191, China }
\email{handr@buaa.edu.cn}

%

\author{Meng Huang}
\address{School of Mathematical Sciences, Beihang University, Beijing, 100191, China }
\email{menghuang@buaa.edu.cn}


\baselineskip 18pt

\title{The Condition Number in Phase Retrieval from Intensity Measurements}


\begin{abstract}
This paper investigates the stability of phase retrieval by analyzing the condition number of the nonlinear map \(\Psi_{\ma}(\bm{x}) = \bigl(\lvert \langle \va_j, \vx \rangle \rvert^2 \bigr)_{1 \le j \le m}\), where \(\va_j \in \mathbb{H}^n\) are known sensing vectors with \(\mathbb{H} \in \{\mathbb{R}, \mathbb{C}\}\). For each \(p \ge 1\), we define 
the condition number \(\beta_{\Psi_\ma}^{\ell_p}\)  as the ratio of optimal upper and lower Lipschitz constants  of \(\Psi_{\ma}\) measured in the \(\ell_p\) norm,  with respect to the metric \(\disth{\vx, \vy} = \|\vx \vx^\ast - \vy \vy^\ast\|_*\).  We establish universal lower bounds on \(\beta_{\Psi_\ma}^{\ell_p}\) for any sensing matrix \(\ma \in \mathbb{H}^{m \times d}\),  proving  that \(\beta_{\Psi_\ma}^{\ell_1} \ge \pi/2\) and \(\beta_{\Psi_\ma}^{\ell_2} \ge \sqrt{3}\) in the real case \((\mathbb{H} = \mathbb{R})\), and \(\beta_{\Psi_\ma}^{\ell_p} \ge 2\) for \(p=1,2\) in the complex case \((\mathbb{H} = \mathbb{C})\). These bounds are shown to be asymptotically tight: both a deterministic harmonic frame \(\mem \in \mathbb{R}^{m \times 2}\)  and Gaussian random matrices \(\ma \in \mathbb{H}^{m \times d}\) asymptotically attain them. Notably, the harmonic frame \(\mem \in \mathbb{R}^{m \times 2}\) achieves the optimal lower bound \(\sqrt{3}\) for all \(m \ge 3\) when \(p=2\), thus serving as an optimal sensing matrix within  \(\ma \in \mathbb{R}^{m \times 2}\).  Our results provide the first explicit uniform lower bounds on \(\beta_{\Psi_\ma}^{\ell_p}\) and offer insights into the fundamental stability limits of phase retrieval.
\end{abstract}

\maketitle

\section{Introduction}
\subsection{Phase retrieval}
This paper addresses the phase retrieval problem, which aims to recover a signal $\vx \in \H^d$ from the intensity measurements $\abs{\nj{\va_j,\vx}}^2, j=1,\ldots,m$, where $ \va_j \in \H^n$ are known sensing vectors and  $\H\in\dkh{\R,\C}$.  This problem naturally arises in applications such as X-ray crystallography \cite{harrison1993phase,millane1990phase} and coherent diffraction imaging \cite{chai2010array,shechtman2015phase}, where only the magnitudes of measurements are accessible. The practical significance of phase retrieval has spurred extensive mathematical research, particularly within frame theory; for comprehensive surveys on this topic, we refer the reader to \cite{grohs2020phase}. For convenience, we denote $\ma:=(\va_1,\ldots,\va_m)^\hh \in \H^{m\times d}$ and consider the nonlinear map
 \begin{equation*}
	\Psi_{\ma}(\vx) = \abs{\ma\vx}^2 := (\abs{\nj{\va_1,\vx}}^2, \abs{\nj{\va_2,\vx}}^2, \cdots, \abs{\nj{\va_m,\vx}}^2). 
\end{equation*}
Then phase retrieval is equivalent to recover \(\vx \in \mathbb{H}^d\) from \(\Psi_{\ma}(\vx) \in \mathbb{R}^m\). Since \(\Psi_{\ma}(c \vx) = \Psi_{\ma}(\vx)\) for any \(c \in \mathbb{H}\) with \(|c| = 1\), the recovery can only be achieved up to a unimodular constant. 
Accordingly, we define an equivalence relation \(\vx \sim \vy\) if and only if there exists \(c \in \mathbb{H}\) with \(|c| = 1\) such that \(\vx = c \vy\). A matrix \(\ma \in \mathbb{H}^{m \times d}\) is said to have the \emph{phase retrieval property} if the map \(\Psi_{\ma}\) is injective on the quotient space \(\mathbb{H}^d / \sim\). In the real case \((\mathbb{H} = \mathbb{R})\), it is known that no matrix \(\ma \in \mathbb{R}^{m \times d}\) with \(m < 2d - 1\) has this property, while a generic matrix with \(m \geq 2d - 1\) has phase retrieval property \cite{balan2006signal}. In the complex case \((\mathbb{H} = \mathbb{C})\), a generic matrix \(\ma \in \mathbb{C}^{m \times d}\) with \(m \geq 4d - 4\) suffices for phase retrieval, whereas no matrix with \(m < 4d - 4\) allows phase retrieval for dimensions \(d = 2^k + 1\) \cite{conca2015algebraic, wanggang2019generalized}.  A more in-depth account of the history of necessary and sufficient bounds for phase retrieval can be found in \cite{botelho2016phase}.

A considerable amount of research has focused on algorithms for phase retrieval, including convex relaxation methods \cite{candes2014phaselift, candes2013phaselift, goldstein2018, waldspurger2015} and non-convex approaches \cite{caijianfeng2022solving, candes2015wf, chenyuxin2015twf, chenyuxin2019, duchi2019solving, huangmeng2022, macong2020, sunju2018, tan2019phase, waldspurger2018phase, wanggang2017taf, zhanghuishuai2016provable, zhanghuishuai2017rwf}, with many providing probabilistic guarantees of global convergence under Gaussian random measurements. Additionally, several works \cite{chenyuxin2015convex,krahmer2018,penghaiyang2024npr} have extended phase retrieval analysis to subgaussian measurement distributions satisfying a small ball condition or to signals meeting a \(\mu\)-flatness assumption when measurements are Bernoulli random vectors. For comprehensive surveys covering recent theoretical and algorithmic advances, we refer readers to \cite{jaganathan2016phase,shechtman2015phase}. In contrast, this paper does not address algorithmic development but instead focuses on the stability of the recovery map \(\Psi_{\ma}(\vx) \to \vx/\sim\).

\subsection{Stability of  phase retrieval}
In phase retrieval applications, the stability of reconstruction is arguably the most critical consideration. By analyzing the bi-Lipschitz properties of the map \(\Psi_{\ma}\) and  its variants, it has been shown that phase retrieval in finite-dimensional settings is uniformly stable \cite{alaifari2017phase,bandeira2014saving,cahill2016phase,grohs2020phase}.  Throughout this paper, we say that  the map \(\Psi_{\ma}\) is bi-Lipschitz if  there exist constants $0 < L \le U < \infty$ such that for all  $\vx, \vy \in \H^d$, 
\begin{equation} \label{def:lowupcon}
	L\cdot \disth{\vx,\vy}\le \ucnorms{\Psi_{\ma}(\vx)- \Psi_{\ma}(\vy)}_p \le U\cdot \disth{\vx, \vy}, 
\end{equation}
where $p\ge 1$ and 
 \begin{equation*}
	\disth{\vx,\vy} := \norms{\vx\vx^\hh-\vy\vy^\hh}_*  \overset{\text{(i)}}{=}   \inf_\theta\ucnorm{\vx-\e^{\i\theta} \vy} \cdot \sup_\theta\ucnorm{\vx-\e^{\i\theta} \vy}.
\end{equation*}
Here, $\norms{\cdot}_*$ denotes the nuclear norm, and the equality (i) is established in  \cite[Theorem 3.7]{balan2022lipschitz}. 
 In particular, when the measurement vectors $\va_j, j=1,\ldots,m$ are subgaussian random vectors satisfying small-ball assumption, the authors of \cite{davis2020nonsmooth, duchi2019solving, eldar2014phase} focus on the lower Lipschitz constant and provide a deterministic constant \( L> 0 \) such that
 \begin{equation*}
	\ucnorms{\Psi_{\ma}(\vx)- \Psi_{\ma}(\vy)}_1 \ge L \cdot \inf_\theta\ucnorm{\vx-\e^{\i\theta} \vy} \cdot \sup_\theta\ucnorm{\vx-\e^{\i\theta} \vy}, \quad \mbox{for all} \quad \vx,\vy\in\H^d.
\end{equation*} 
For a given measurement matrix, one can increase the lower Lipschitz constant \(L\) by amplifying the measurement vectors, thereby enhancing \(\Psi_{\ma}\)'s ability to distinguish signals and improving noise robustness; however, such amplification can be costly \cite{bandeira2014saving}. This motivates the study of stability via the condition number, which remains invariant under scaling of the measurement vectors. Formally, for any \(\ma \in \mathbb{H}^{m \times d}\) and \(p \geq 1\), the condition number of \(\Psi_{\ma}\) with respect to the \(\ell_p\) norm is defined as
\begin{equation}\label{def:beta}
\beta^{\ell_p}_{\Psi_\ma} = \frac{U^{\ell_p}_{\Psi_\ma}}{L^{\ell_p}_{\Psi_\ma}},
\end{equation}
where \(L^{\ell_p}_{\Psi_\ma}\) and \(U^{\ell_p}_{\Psi_\ma}\) are the optimal lower and upper Lipschitz constants from \eqref{def:lowupcon}, given by
\begin{equation}\label{def:l and u}
	L^{\ell_p}_{\Psi_\ma} = \inf_{\substack{\vx,\vy\in\H^d\\ \disth{\vx, \vy}\ne 0}} \frac{\ucnorms{\Psi_{\ma}(\vx)- \Psi_{\ma}(\vy)}_p}{\disth{\vx, \vy}},\quad 
	U^{\ell_p}_{\Psi_\ma} = \sup_{\substack{\vx,\vy\in\H^d\\ \disth{\vx, \vy}\ne 0}} \frac{\ucnorms{\Psi_{\ma}(\vx)- \Psi_{\ma}(\vy)}_p}{\disth{\vx, \vy}}. 
\end{equation}
The condition number serves as a quantitative measure of the stability of a measurement matrix \(\ma\) for phase retrieval under the \(\ell_p\) norm. It is theoretically established that the smaller the condition number, the greater the stability of \(\Psi_{\ma}\) against adversarial noise \cite{bandeira2014saving}. If \(\ma\) lacks the phase retrieval property, then the lower Lipschitz bound satisfies \(L^{\ell_p}_{\Psi_\ma} = 0\), implying \(\beta^{\ell_p}_{\Psi_\ma} = +\infty\). Moreover, by Theorem \ref{th:equivalent}, stated later, the upper Lipschitz bound satisfies \(U_{\Psi_\ma}^{\ell_p} = \|\ma\|_{2 \to 2p}^2 < +\infty\). Consequently, if \(\ma\) has the phase retrieval property, the condition number \(\beta^{\ell_p}_{\Psi_\ma}\) is finite \cite{grohs2020phase} and satisfies \(\beta^{\ell_p}_{\Psi_\ma} \geq 1\). In this paper, we focus on analyzing phase retrieval stability through the condition number of the nonlinear map \(\Psi_{\ma}\).

A more granular analysis was conducted by Balan and Zou \cite{balan2016lipschitz}, who investigated the local Lipschitz properties of both the intensity map \(\Psi_{\ma}\) and the magnitude map \(\Phi_{\ma} = \bigl(|\langle \va_1, \vx \rangle|, \ldots, |\langle \va_m, \vx \rangle|\bigr)\). They defined distinct types of local Lipschitz bounds and derived explicit expressions for them in both real and complex settings. This work highlighted that stability can vary significantly depending on the specific location within the signal space.

The choice of metric for analyzing stability is particularly crucial in the context of generalized phase retrieval respect to the low-rank matrix. Balan and Dock \cite{balan2022lipschitz} studied the Lipschitz properties of the generalized maps \(\Psi_{\ma}\) and \(\Phi_{\ma}\) using various metrics on the quotient space, including the metric \(\disth{\vx, \vy}\) and the Bures-Wasserstein distance \( \dd(\vx,\vy) = \inf_\theta\norm{\vx-\e^{\i\theta} \vy} \). Their analysis revealed a critical distinction: the intensity map \(\Psi_{\ma}\) remains globally lower Lipschitz with respect to the Frobenius norm, ensuring its stability; in stark contrast, the magnitude map \(\Phi_{\ma}\) was shown to fail this stability test, with its global lower Lipschitz constant being zero for some case with the Bures-Wasserstein distance. This suggests that the intensity map \(\Psi_{\ma}\) is a more robust and natural formulation for generalized phase retrieval problems.

\subsection{Our contributions}
Calculating the condition number for phase retrieval is notoriously challenging. In this paper, we first provide an alternative characterization of the optimal lower and upper Lipschitz bounds \(L^{\ell_p}_{\Psi_\ma}\) and \(U^{\ell_p}_{\Psi_\ma}\). Specifically, for any sensing matrix \(\ma = (\va_1, \ldots, \va_m)^\ast \in \mathbb{H}^{m \times d}\) and every \(p \geq 1\), we prove that
\begin{align} \label{eq:charbu}
L_{\Psi_\ma}^{\ell_p} &= \inf_{\substack{\|\vu\|_2 = \|\vv\|_2 = 1 \\ \langle \vu, \vv \rangle \in \mathbb{R}}} \left( \sum_{j=1}^m \left| \operatorname{Re}(\vu^\ast \va_j \va_j^\ast \vv) \right|^p \right)^{\frac{1}{p}}, \quad \text{and} \quad
U_{\Psi_\ma}^{\ell_p} = \|\ma\|_{2 \to 2p}^2,
\end{align}
where \(\|\ma\|_{2 \to 2p} := \sup_{\|\vx\|=1} \|\ma \vx\|_{2p}\). Furthermore, for \(p=2\), when \(\{\va_j\}_{j=1}^m\) forms a tight 4-frame, the minimum  is attained by orthogonal vectors \(\vu\) and \(\vv\), that is,
\[
L_{\Psi}^{\ell_2} = \inf_{\substack{\|\vu\|_2 = \|\vv\|_2 = 1 \\ \langle \vu, \vv \rangle = 0}} \left( \sum_{j=1}^m \left| \operatorname{Re}(\vu^\ast \va_j \va_j^\ast \vv) \right|^2 \right)^{\frac{1}{2}}.
\]
This significantly simplifies the computation of the condition number by restricting the search to orthogonal pairs \(\vu, \vv \in \mathbb{H}^d\) rather than all pairs with real inner product.

Building on the characterization \eqref{eq:charbu}, we establish uniform lower bounds on the condition number \(\beta^{\ell_p}_{\Psi_\ma}\) for any sensing matrix \(\ma \in \mathbb{H}^{m \times d}\) and for both \(p=1\) and \(p=2\), as detailed in Lemmas \ref{le:opt:number:2} and \ref{le:opt:number:l1}. Specifically, we prove that
\begin{align} \label{eq:unibcon}
\beta^{\ell_1}_{\Psi_\ma} \geq   \begin{cases}
m \tan\frac{\pi}{2m}, & \text{if } \mathbb{H} = \mathbb{R}, \\
2, & \text{if } \mathbb{H} = \mathbb{C},
\end{cases}
\quad \text{and} \quad
\beta^{\ell_2}_{\Psi_\ma} \geq   \begin{cases}
\sqrt{3}, & \text{if } \mathbb{H} = \mathbb{R}, \\
\ \; 2, & \text{if } \mathbb{H} = \mathbb{C}.
\end{cases}
\end{align}
Generally speaking, a smaller condition number indicates that \(\Psi_{\ma}\) more closely approximates an isometric mapping. Our results \eqref{eq:unibcon} reveal that the condition number for phase retrieval can never reach 1, regardless of the dimensions \(m\) and \(d\), in contrast to linear maps where a large class of matrices achieve condition number exactly equal to 1.

Finally, we demonstrate that the bounds in \eqref{eq:unibcon} are tight or asymptotically tight by computing the condition numbers of the harmonic frame \(\mem \in \mathbb{R}^{m \times 2}\) and Gaussian random matrices \(\ma \in \mathbb{H}^{m \times d}\). Specifically, for any \(m \geq 3\), the condition numbers of the harmonic frame \(\mem\) in \(\mathbb{R}^2\) satisfy
 \begin{align} \label{eq:harcon}
	\betaone_{\mem} = \left\{\begin{aligned}
		&  m \frac{\tan \pimtwo}{\cos \pimtwo}, \quad \mbox{if}\ m\ \mbox{is odd},\\
		& \frac m2 \tan \frac{\pi}{m},  \quad \mbox{if}\ m\ \mbox{is even},
	\end{aligned}\right. \qquad \mbox{and} \qquad \beta_\mem^{\ell_2} = \sqrt{3}.
\end{align}
Here, the harmonic frame is defined as
\begin{equation} \label{def:mem}
\mem := \begin{pmatrix}
1 & \cos\left(\frac{\pi}{m}\right) & \cdots & \cos\left(\frac{(m-1)\pi}{m}\right) \\
0 & \sin\left(\frac{\pi}{m}\right) & \cdots & \sin\left(\frac{(m-1)\pi}{m}\right)
\end{pmatrix}^\top \in \mathbb{R}^{m \times 2}.
\end{equation}
These results indicate that \(\mem\) is the optimal sensing matrix minimizing the condition number in \(\mathbb{R}^{m \times 2}\) for \(p=2\), and asymptotically optimal for \(p=1\) as \(m \to \infty\). Moreover, for any fixed \(0 < \delta < 1\), with high probability the condition number of a standard Gaussian matrix \(\ma \in \mathbb{H}^{m \times d}\) satisfies
\[
\bar{\beta}^{\ell_1} \leq \beta^{\ell_1}_{\Psi_\ma} \leq \bar{\beta}^{\ell_1} + \delta,
\]
provided \(m \geq C_0 \log(1/\delta) \delta^{-2} d\) for some constant \(C_0 > 0\) where \(\bar{\beta}^{\ell_1} = \pi/2\) for \(\H = \R\) and \(\bar{\beta}^{\ell_1} = 2\) for \(\H = \C\).   
Similar asymptotic tightness holds for \(p=2\) as \(m \to \infty\). Together, these results establish that the universal lower bounds \(\bar{\beta}^{\ell_1}\) and \(\bar{\beta}^{\ell_2}\) in \eqref{eq:unibcon} are asymptotically tight in \(\mathbb{H}^{m \times d}\).

\subsection{Comparison to previous work}
Another widely used nonlinear map in phase retrieval is 
\begin{equation*}
	\Phi_{\ma}(\vx) = |\ma \vx| := \bigl(|\langle \va_1, \vx \rangle|, |\langle \va_2, \vx \rangle|, \ldots, |\langle \va_m, \vx \rangle|\bigr),
\end{equation*}
where \(\va_j \in \mathbb{H}^n\) for \(j=1,\ldots,m\) are known sensing vectors. The stability of \(\Phi_{\ma}\) has been extensively studied through its bi-Lipschitz property, especially focusing on the lower Lipschitz bound \cite{alaifari2017phase, alharbi2024locality, balan2015invertibility, bandeira2014saving, cahill2016phase}. Specifically, using the Bures-Wasserstein distance defined by \( \dd(\vx,\vy) = \inf_\theta\norm{\vx-\e^{\i\theta} \vy} \),
we say \(\Phi_{\ma}\) satisfies the bi-Lipschitz property with constants \(U \geq L > 0\) if for all \(\vx, \vy \in \mathbb{H}^d\),
\[
 L  \cdot \dd(\vx,\vy)  \le \norm{\Phi_{\ma}(\vx)-\Phi_{\ma}(\vy)} \le  U  \cdot \dd(\vx,\vy).
\]
The optimal lower and upper Lipschitz bounds are denoted by \(L_{\Phi_\ma}\) and \(U_{\Phi_\ma}\), respectively.
It has been shown that $U_{\Phi_\ma}=\norm{\ma}$. Furthermore, for the real case,  it is known that $\sigma_{\ma} \le L_{\Phi_\ma} \le \sqrt2 \sigma_{\ma} $, where $\sigma_{\ma} := \min_{S \subset [m]} \max\dkh{\lambda_{\min}(A_S^\hh A_{S} ), \lambda_{\min}(A_{S^c}^\hh A_{S^c} )}$ and \(\lambda_{\min}\) denotes the minimal eigenvalue. In the complex case, however, only $L_{\Phi_\ma} \le C \sigma_{\ma} $
can be established for some constant \(C > 0\) depending on \(\ma\) \cite{alaifari2017phase}.

Recently, Xia et al.~\cite{xiayu2024stability} studied the condition number of  \(\Phi_{\ma}\), defined as
\begin{equation}\label{eq:beta:1}
\beta_{\Phi_\ma} = \frac{U_{\Phi_\ma}}{L_{\Phi_\ma}}.
\end{equation}
They established uniform lower bounds on \(\beta_{\Phi_\ma}\) for all \(\ma \in \mathbb{H}^{m \times d}\) in both real and complex settings (see Table~\ref{tab:para of cond}), and demonstrated that these bounds are asymptotically tight. A key motivation for our study of \(\Psi_{\ma}\) is that in many practical scenarios, noise is added directly to \(\Psi_{\ma}\) rather than to \(\Phi_{\ma}\). Table~\ref{tab:para of cond} compares the condition numbers of \(\Psi_{\ma}\) obtained in this paper with those of \(\Phi_{\ma}\) from \cite{xiayu2024stability}.
\begin{table}[H]
	\renewcommand\arraystretch{1.5}
	\caption{Uniform lower bounds of the condition numbers of $\Psi_{\ma}$ and $\Phi_{\ma}$.}\label{tab:para of cond}
	\begin{center}
		\begin{tabular}{|c|c|c|c|}\hline
			Non-linear map            & $\H = \R$          & $\H = \C$     &       \\\hline
			$\betaone_{\Psi_\ma} $    &$\frac{\pi}{2}$     & $2$           &  this paper   \\
			$\betatwo_{\Psi_\ma} $    &$\sqrt{3}$          & $2$           &  this paper   \\
			\hspace{-0.5em} $\beta_{\Phi_\ma}$  &$\sqrt{\frac{\pi}{\pi-2}} \approx 1.659$& $\sqrt{\frac{4}{4-\pi}}\approx 2.159$ & \cite{xiayu2024stability} \\\hline
		\end{tabular}
	\end{center}
\end{table}
 It should be noted that  our study  focuses exclusively on the stability of phase retrieval in finite-dimensional spaces \(\mathbb{H}^d\). More general settings, including phase retrieval in infinite-dimensional Banach spaces and their stability properties, are discussed in \cite{alaifari2017phase, cahill2016phase, grohs2020phase}. Unlike the finite-dimensional case, it has been shown that phase retrieval in infinite dimensions is inherently unstable. Additionally, the stability of finite-dimensional phase retrieval under perturbations of the frame vectors has been studied in \cite{alharbi2023stable}.

\subsection{Notations} 
For any vector \(\vx\), \(\|\vx\|_p\) denotes its \(\ell_p\)-norm. The symbols \(\vx^\top\) and \(\vx^\ast\) represent the transpose and conjugate transpose of \(\vx\), respectively. For a matrix \(\ma\), \(\|\ma\|\), \(\|\ma\|_*\), and \(\|\ma\|_F\) denote its spectral norm, nuclear norm, and Frobenius norm, respectively, while \(\|\ma\|_{2 \to 2p} := \sup_{\|\vx\|=1} \|\ma \vx\|_{2p}\). The notations \(\ma^\top\), \(\ma^\ast\), and \(\overline{\ma}\) stand for the transpose, conjugate transpose, and entrywise conjugate of \(\ma\), respectively. We write \(\operatorname{Re}(\cdot)\) and \(\operatorname{Im}(\cdot)\) for the real and imaginary parts of a complex number or matrix (applied entrywise). For \(\mathbb{H} \in \{\mathbb{R}, \mathbb{C}\}\), \(\mathbb{S}_{\mathbb{H}}^{d-1} := \{\vx \in \mathbb{H}^d : \norm{\vx} = 1\}\) denotes the unit sphere. A vector \(\va \in \mathbb{H}^d\) is called a standard Gaussian random vector if
\begin{equation} \label{def:gauss}
\va \sim \begin{cases}
\mathcal{N}(\mathbf{0}, \mi_d), & \text{if } \mathbb{H} = \mathbb{R}, \\
\mathcal{N}(\mathbf{0}, \mi_d/2) + \i\, \mathcal{N}(\mathbf{0}, \mi_d/2), & \text{if } \mathbb{H} = \mathbb{C},
\end{cases}
\end{equation}
where \(\mathcal{N}(\mathbf{0}, \Sigma)\) denotes a Gaussian distribution with mean zero and covariance \(\Sigma\). A matrix \(\ma \in \mathbb{H}^{m \times d}\) is a standard Gaussian random matrix if its rows are independent standard Gaussian vectors. We use the notation \(f(m,n) = O(g(m,n))\) or \(f(m,n) \lesssim g(m,n)\) to mean there exists a constant \(c > 0\) such that \(f(m,n) \leq c\, g(m,n)\), and \(f(m,n) \gtrsim g(m,n)\) to mean \(f(m,n) \geq c\, g(m,n)\) for some \(c > 0\). Finally, \([m] := \{1, 2, \ldots, m\}\) for any positive integer \(m\).

\subsection{Organization}
The remainder of this paper is organized as follows: Section~\ref{sec:preliminaries} introduces equivalent characterizations of the optimal Lipschitz constants \(L_{\Psi_\ma}^{\ell_p}\) and \(U_{\Psi_\ma}^{\ell_p}\), along with orthogonality conditions as preliminaries. In Section~\ref{sec:p = 2}, we focus on the condition number with respect to the \(\ell_2\) norm, establishing optimal uniform lower bounds, concentration results for random matrices, and showing that the harmonic frame attains the optimal lower bound when \(\mathbb{H} = \mathbb{R}\). Section~\ref{sec:p = 1} presents analogous results for the \(\ell_1\) norm. Finally, Section~\ref{sec:discussion} provides a discussion and outlook. Technical lemmas and detailed proofs are collected in the supplementary materials.

\section{Characterizing the Lipschitz constants: Equivalent forms and orthogonal conditions}\label{sec:preliminaries}
In this section, we first present alternative expressions for \(L_{\Psi_\ma}^{\ell_p}\) and \(U_{\Psi_\ma}^{\ell_p}\), and then investigate the conditions under which the  lower Lipschitz bound is attained at orthogonal vectors. These results aim to simplify the characterization of \(L_{\Psi_\ma}^{\ell_p}\) and \(U_{\Psi_\ma}^{\ell_p}\), thereby facilitating the computation.

\subsection{Others expressions for \texorpdfstring{$L_{\Psi}^{\ell_p}$ and $U_{\Psi}^{\ell_p}$.}{} } 
The following theorem extends \cite[Theorem 4.5]{balan2015invertibility} from the real case to the complex case.

\begin{theorem}\label{th:equivalent}
	For any matrix $\ma = (\va_1,\cdots,\va_m)^\hh \in \H^{m\times d}$ and each $p\ge 1$, one has
	 \begin{align*}
		L_{\Psi_\ma}^{\ell_p} = \inf_{\substack{\norm{\vu}=\norm{\vv}=1 \\ \nj{\vu,\vv}\in\R}} \xkh{\sum_{j=1}^m \abs{\real{\vu^\hh\va_j \va_j^\hh\vv}}^p}^{\frac 1p} \quad \mbox{and} \quad U_{\Psi_\ma}^{\ell_p} =  \norms{\ma}_{2\rightarrow 2p}^2.
\end{align*}
\end{theorem}
\begin{proof}
	For any $\vx, \vy \in \H^d$ with $\disth{\vx, \vy}\ne 0$, let $\theta_0 = \argmin{\theta\in\R} \norm{\vx-\e^{\i\theta}\vy}$.
It is easy to verify that $\nj{\vx, \e^{\i\theta_0} \vy} \in \R $ and
\[  \maxm{\theta\in\R} \ucnorm{\vx-\e^{\i\theta}\vy}=\ucnorm{\vx+\e^{\i\theta_0}\vy}. \] 
Define $\vu = \vx-\e^{\i\theta_0}\vy, \vv = \vx+\e^{\i\theta_0}\vy$.  Then $0<\norm{\vu} \le \norm{\vv}$ and $\disth{\vx, \vy}=\norm{\vu} \norm{\vv}$.  Moreover,
\[\nj{\vu, \vv} = \norm{\vx}^2 - \norm{\vy}^2 + 2 \image{\nj{\vx, \e^{\i\theta_0} \vy}} \cdot \i = \norm{\vx}^2 - \norm{\vy}^2 \in \R .
\]
A straightforward calculation yields
\begin{equation} \label{eq:reuvxy}
\ucabs{\va_j^\hh \vx}^2 - \ucabs{\va_j^\hh \vy}^2=\real{\vu^\hh\va_j \va_j^\hh\vv}.
\end{equation}
Therefore,  by the definition of \(L_{\Psi_\ma}^{\ell_p}\) in \eqref{def:l and u}, we have
\[
L^{\ell_p}_{\Psi_\ma} \ge   \inf_{\substack{\vu,\vv\in\H^d, \nj{\vu,\vv}\in\R\\  \vu,\vv \ne 0}} \frac{\xkh{\sum_{j=1}^m \abs{\real{\vu^\hh\va_j \va_j^\hh\vv}}^p}^{\frac 1p}}{\norm{\vu}\norm{\vv}} = \inf_{\substack{\norm{\vu}=\norm{\vv}=1 \\ \nj{\vu,\vv}\in\R}} \xkh{\sum_{j=1}^m \abs{\real{\vu^\hh\va_j \va_j^\hh\vv}}^p}^{\frac 1p}. 
\]
Conversely, for any $\vu,\vv\in\H^d$ with  $ \nj{\vu,\vv}\in\R$ and $\norm{\vu}=\norm{\vv}=1$, let $\vx=(\vv+\vu)/2$ and $\vy=(\vv-\vu)/2$. Then $\disth{\vx, \vy}=1$, and using \eqref{eq:reuvxy} gives
\[
\inf_{\substack{\norm{\vu}=\norm{\vv}=1 \\ \nj{\vu,\vv}\in\R}} \xkh{\sum_{j=1}^m \abs{\real{\vu^\hh\va_j \va_j^\hh\vv}}^p}^{\frac 1p}\ge  \inf_{\substack{\vx,\vy\in\H^d\\ \disth{\vx, \vy}\ne 0}} \frac{\ucnorms{\abs{\ma\vx}^2-\abs{\ma\vy}^2}_p}{\disth{\vx, \vy}} =L^{\ell_p}_{\Psi_\ma}.
\]
Combining the above inequalities, we conclude
\begin{align*}
	L_{\Psi_\ma}^{\ell_p}  = \inf_{\substack{\norm{\vu}=\norm{\vv}=1 \\ \nj{\vu,\vv}\in\R}} \xkh{\sum_{j=1}^m \abs{\real{\vu^\hh\va_j \va_j^\hh\vv}}^p}^{\frac 1p}. 
\end{align*}
Similarly, the upper Lipschitz constant \(U_{\Psi_\ma}^{\ell_p}\) satisfies
\begin{align*}
	U_{\Psi_\ma}^{\ell_p} = \sup_{\substack{\vu,\vv\in\H^d,  \nj{\vu,\vv}\in\R \\ \vu,\vv \neq 0}} \frac{\xkh{\sum_{j=1}^m \abs{\real{\vu^\hh\va_j \va_j^\hh\vv}}^p}^{\frac 1p}}{\norm{\vu}\norm{\vv}} = \sup_{\norm{\vu}=1} \xkh{\sum_{j=1}^m \abs{\va_j^\hh\vu}^{2p}}^{\frac 1p},
\end{align*}
where the last equality follows from the Cauchy–Schwarz inequality.  
\end{proof}

\begin{remark}
By the definition of \(L_{\Psi_\ma}^{\ell_p}\) in \eqref{def:l and u}, a matrix \(\ma \in \mathbb{H}^{m \times d}\) has the phase retrieval property if and only if \(L_{\Psi_\ma}^{\ell_p} > 0\). In particular, for \(p=2\), Theorem \ref{th:equivalent} implies that \(\ma\) has phase retrieval if and only if there exists a constant \(c_0 > 0\) such that
\begin{equation} \label{eq:uniel2}
\sum_{j=1}^m \left| \operatorname{Re}(\vu^\ast \va_j \va_j^\ast \vv) \right|^2 \geq c_0 \|\vu\|_2^2 \|\vv\|_2^2
\end{equation}
holds for all \(\vu, \vv \in \mathbb{H}^d\) with \(\langle \vu, \vv \rangle \in \mathbb{R}\). Previous work \cite{balan2015stability} established that \(\ma\) has the phase retrieval property if and only if there exists \(c_0 > 0\) such that for all \(\vx, \vy \in \mathbb{H}^d\),
\[
\|\Psi_{\ma}(\vx) - \Psi_{\ma}(\vy)\|_2^2 \geq c_0 \left( \|\vx - \vy\|_2^2 \|\vx + \vy\|_2^2 - 4\, \operatorname{Im}^2 \langle \vx, \vy \rangle \right).
\]
Compared to this, the condition in \eqref{eq:uniel2} is more straightforward and convenient to verify.

\end{remark}

According to Theorem \ref{th:equivalent}, a straightforward conclusion is that the minimal lower Lipschitz bound in \eqref{def:l and u} is attained when \(\vx\) and \(\vy\) are orthogonal, as stated below.

\begin{corollary}\label{th:ortho:1}
	For any matrix $\ma = (\va_1,\cdots,\va_m)^\hh \in \H^{m\times d}$ and each $p \ge 1$, the lower Lipschitz constant $L^{\ell_p}_{\Psi_\ma}$ satisfies 
	\begin{align*}
			L^{\ell_p}_{\Psi_\ma} = \inf_{\substack{\vx,\vy\in\H^d, \nj{\vx,\vy}=0\\ \norm{\vx} \le 1,\norm{\vy} = 1}} \frac{\ucnorms{\abs{\ma\vx}^2-\abs{\ma\vy}^2}_p}{\disth{\vx, \vy}}.
		\end{align*}
\end{corollary}
\begin{proof}
For any \(\vu, \vv \in \mathbb{H}^d\) with \(\langle \vu, \vv \rangle \in \mathbb{R}\) and \(\|\vu\| = \|\vv\| = 1\), define \(\vx = (\vv + \vu)/2\) and \(\vy = (\vv - \vu)/2\). It is straightforward to verify that \(\langle \vx, \vy \rangle = 0\) and  $\disth{\vx, \vy}=1\ne 0$. Note that
\[
\operatorname{Re}(\vu^\ast \va_j \va_j^\ast \vv) = |\va_j^\ast \vx|^2 - |\va_j^\ast \vy|^2.
\]
Therefore,
\[
 L^{\ell_p}_{\Psi_\ma}= \inf_{\substack{\norm{\vu}=\norm{\vv}=1 \\ \nj{\vu,\vv}\in\R}} \xkh{\sum_{j=1}^m \abs{\real{\vu^\hh\va_j \va_j^\hh\vv}}^p}^{\frac 1p}  \ge  \inf_{\substack{\vx,\vy\in\H^d ,  \nj{\vx,\vy}=0 \\ \disth{\vx, \vy} \neq 0}} \frac{\ucnorms{\abs{\ma\vx}^2-\abs{\ma\vy}^2}_p}{\disth{\vx, \vy}} ,
\]
where the equality on the left follows from Theorem \ref{th:equivalent}. Clearly,
\[
 \inf_{\substack{\vx,\vy\in\H^d ,  \nj{\vx,\vy}=0 \\ \disth{\vx, \vy} \neq 0}} \frac{\ucnorms{\abs{\ma\vx}^2-\abs{\ma\vy}^2}_p}{\disth{\vx, \vy}} \ge  \inf_{\substack{\vx,\vy\in\H^d\\ \disth{\vx, \vy}\ne 0}} \frac{\ucnorms{\abs{\ma\vx}^2-\abs{\ma\vy}^2}_p}{\disth{\vx, \vy}} = L^{\ell_p}_{\Psi_\ma},
\]
where the last equality is by definition of \(L_{\Psi_\ma}^{\ell_p}\). Combining these inequalities, we conclude that
\[
L_{\Psi_\ma}^{\ell_p} =  \inf_{\substack{\vx,\vy\in\H^d ,  \nj{\vx,\vy}=0 \\ \disth{\vx, \vy} \neq 0}} \frac{\ucnorms{\abs{\ma\vx}^2-\abs{\ma\vy}^2}_p}{\disth{\vx, \vy}}.
\]
\end{proof}

\begin{remark}
A similar result has been established for the nonlinear map $\Phi$ in \cite[Theorem 1.1]{alharbi2024locality}, where the minimal lower Lipschitz bound of $\Phi$ occurs at orthogonal vectors.
\end{remark}

\subsection{Orthogonal conditions}
Recall from Theorem \ref{th:equivalent} that the lower Lipschitz bound is given by
\[
L_{\Psi_\ma}^{\ell_p} = \inf_{\substack{\|\vu\|_2 = \|\vv\|_2 = 1 \\ \langle \vu, \vv \rangle \in \mathbb{R}}} \left( \sum_{j=1}^m \left| \operatorname{Re}(\vu^\ast \va_j \va_j^\ast \vv) \right|^p \right)^{\frac{1}{p}}.
\]
To further simplify the computation, it is natural to ask under what conditions the infimum is attained when \(\vu\) and \(\vv\) are orthogonal. The next lemma provides a sufficient condition for this to hold in the case \(p=2\).

\begin{lemma}\label{le:ortho:2:p=2}
	Let $\ma = (\va_1,\cdots,\va_m)^\hh \in \H^{m\times d}$. If $\sum_{j=1}^{m} \ucabs{\va_j^\hh \vx}^4 = \sum_{j=1}^{m} \ucabs{\va_j^\hh \vy}^4 $ holds for all $\vx, \vy \in\H^d$ with $\norm{\vx} = \norm{\vy} = 1$ and  $\nj{\vx, \vy}=0$, then 
	\begin{align*}
		L_{\Psi_\ma}^{\ell_2} = \inf_{\substack{\norm{\vu}=\norm{\vv}=1 \\ \nj{\vu,\vv} = 0}} \xkh{\sum_{j=1}^m \abs{\real{\vu^\hh\va_j \va_j^\hh\vv}}^2}^{\frac 12}. 
	\end{align*}
\end{lemma}
\begin{proof}
To prove the result, it suffices to show that if the condition
 $\sum_{j=1}^{m} \ucabs{\va_j^\hh \vx}^4 = \sum_{j=1}^{m} \ucabs{\va_j^\hh \vy}^4 $ holds for all $\vx, \vy \in\H^d$ with $\norm{\vx} = \norm{\vy} = 1$ and $\nj{\vx, \vy}=0$, then 
	\begin{equation} \label{eq:normorth}
		L_{\Psi_\ma}^{\ell_2} =  \inf_{\substack{\vx,\vy\in\H^d ,  \nj{\vx,\vy}=0 \\ \norm{\vx}=\norm{\vy}=1}} \frac{\ucnorms{\abs{\ma\vx}^2-\abs{\ma\vy}^2}_2}{\disth{\vx, \vy}}.
	\end{equation}
Indeed, for any such $\vx, \vy \in\H^d$, define
\[
\vu = \frac{\vx - \vy}{\sqrt{2}}, \quad \vv = \frac{\vx + \vy}{\sqrt{2}}.
\]
Then $\norm{\vu}=\norm{\vv}=1$, $ \nj{\vu,\vv}=0$, and  $\disth{\vx, \vy}=2$.  Note that $\ucabs{\va_j^\hh \vx}^2 - \ucabs{\va_j^\hh \vy}^2=2\ucreal{\vu^\hh\va_j \va_j^\hh\vv} $. Therefore, 
\begin{equation} \label{eq:gexorthy}
 \inf_{\substack{\vx,\vy\in\H^d ,  \nj{\vx,\vy}=0 \\ \norm{\vx}=\norm{\vy}=1}} \frac{\ucnorms{\abs{\ma\vx}^2-\abs{\ma\vy}^2}_2}{\disth{\vx, \vy}} \ge \inf_{\substack{\norm{\vu}=\norm{\vv}=1 \\ \nj{\vu,\vv} = 0}} \xkh{\sum_{j=1}^m \abs{\real{\vu^\hh\va_j \va_j^\hh\vv}}^2}^{\frac 12} \ge 	L_{\Psi_\ma}^{\ell_2},
\end{equation}
where the last inequality follows from Theorem \ref{th:equivalent}. Combining \eqref{eq:normorth} and \eqref{eq:gexorthy}, we conclude that
\[
L_{\Psi_\ma}^{\ell_2} = \inf_{\substack{\norm{\vu}=\norm{\vv}=1 \\ \nj{\vu,\vv} = 0}} \xkh{\sum_{j=1}^m \abs{\real{\vu^\hh\va_j \va_j^\hh\vv}}^2}^{\frac 12}. 
\]

Next, we prove \eqref{eq:normorth}.  For any $\vx,\vy \in \H^d$ with $\nj{\vx,\vy}=0$ and $\norm{\vx}=\norm{\vy}=1$, define  $\vx^t = \sqrt{t}\,\vx$ and consider the function
 \begin{equation*}
	g(t) := \frac{\ucnorm{\ucabs{\ma\vxt}-\ucabs{\ma\vy}}^2}{\mbox{dist}_{\H}^2\xkh{\vxt, \vy}}.
\end{equation*}
We claim that 
\begin{equation} \label{cla:g1}
\min_{t \ge 0} g(t) = g(1).
\end{equation}
The equation \eqref{eq:normorth} then follows by combining \eqref{cla:g1} with Corollary \ref{th:ortho:1}.  
To prove the claim \eqref{cla:g1},  observing  that since \(\langle \vx^t, \vy \rangle = 0\), we have
\[ 
\disth{\vxt, \vy}=\ucnorm{\vxt-\vy} \ucnorm{\vxt+\vy}=t+1.
\]
For convenience, set $\alpha_j = \ucabs{\va_j^\hh \vx}^2$, $\gamma_j = \ucabs{\va_j^\hh \vy}^2$. Then
 \begin{equation*}
	g(t) = \frac{\sum_{j=1}^{m} \xkh{\alpha_j t - \gamma_j}^2}{(t+1)^2} .
\end{equation*}
It is straightforward to compute the derivative:
\[
g'(t) = \frac{2 \sumjm (\alpha_j t - \gamma_j)(\alpha_j + \gamma_j)}{(t+1)^3}.
\]
Recall the  assumption $\sum_{j=1}^{m} \ucabs{\va_j^\hh \vx}^4 = \sum_{j=1}^{m} \ucabs{\va_j^\hh \vy}^4 $, which implies 
\[
\sumjm \alpha_j^2 = \sumjm \gamma_j^2.
\] 
Evaluating the derivative at \(t=1\) gives 
\[
g'(1) = \frac{2 \sum_{j=1}^m (\alpha_j - \gamma_j)(\alpha_j + \gamma_j)}{2^3} = \frac{2 \sum_{j=1}^m (\alpha_j^2 - \gamma_j^2)}{8} = 0.
\]
Since the numerator of \(g'(t)\) is linear in \(t\) and all \(\alpha_j, \gamma_j \geq 0\), it follows that \(g(t)\) attains its minimum at \(t=1\). This completes the proof.
\end{proof}

Lemma \ref{le:ortho:2:p=2} requires the condition
\[
\sum_{j=1}^m |\va_j^\ast \vx|^4 = \sum_{j=1}^m |\va_j^\ast \vy|^4
\]
to hold for all \(\vx, \vy \in \mathbb{H}^d\) with \(\|\vx\|_2 = \|\vy\|_2 = 1\) and \(\langle \vx, \vy \rangle = 0\). A sufficient condition for this is that \(\{\va_j\}_{j=1}^m\) forms a tight \(p\)-frame with \(p=4\). Here, we focus on \(p\)-frames in finite-dimensional Hilbert spaces; for a broader treatment of \(p\)-frame theory in Banach spaces, see \cite{casazza2005frame, christensen2003p}.

\begin{definition}\label{def:frame}
	Let $H$ be a finite dimensional Hilbert space and let $1\le p < \infty$. A collection of vectors $\{\va_j\}_{j\in J}$ in $H$ is called a $p$-frame of $H$ with $p$-frame bounds $B \ge A > 0$ if
	 \begin{equation*}
		A \norms{\vx} \le \Big(\sum_{j\in J} \ucabs{\nj{\vx, \va_j}}^p\Big)^{1/p} \le B \norms{\vx} \qquad \mbox{for all} \quad \vx \in H.
	\end{equation*}
If $A = B$, then $\{\va_j\}_{j\in J}$ is called   a tight $p$-frame. 
\end{definition}

It is worth noting that the harmonic frame \(\mem\), defined in \eqref{def:mem}, satisfies $\sum_{j=1}^{m} \lvert\va_j^\hh \vu\rvert^4 = 3m/8$ for all \(\vu \in \mathbb{R}^2\) with \(\norm{\vu} = 1\) (see Section \ref{pf:em:l2}), and is therefore a tight 4-frame. Random frame constructions play a fundamental role in phase retrieval within finite-dimensional Hilbert spaces \cite{alharbi2023stable}. Statistical analysis shows that vectors sampled uniformly at random from the unit sphere in such spaces are highly likely to form an approximate tight frame when the sample size is sufficiently large \cite{vershynin2018}.
 Recent studies on phase retrieval stability have highlighted the advantages of employing \(p\)-norms and their extensions to Banach lattices, with explicit applications discussed in \cite{alaifari2017phase, christ2022examples, freeman2023discretizing}.

Building on Lemma \ref{le:ortho:2:p=2}, one might hope for a similar result to hold when \(p = 1\). However, the following lemma demonstrates that, even for the harmonic frame \(\mem \in \mathbb{R}^{m \times 2}\), the minimal lower Lipschitz bound does not necessarily occur at orthogonal vectors \(\vu\) and \(\vv\).

\begin{lemma}\label{le:ortho:l - em:p=1} 
	Let $m \ge 3$. For $\mem = (\va_1,\cdots,\va_m)^\T \in \R^{m\times 2}$ defined in \eqref{def:mem}, one has: \begin{equation}\label{eq:ortho:1}
		\inf_{\substack{\norm{\vu}=\norm{\vv}=1 \\ \nj{\vu,\vv} = 0}} \sum_{j=1}^m \ucabs{\vu^\T\va_j \va_j^\T\vv} = \left\{\begin{aligned}
			&\frac{1}{2\tan \frac{\pi}{2 m}}, \quad \mbox{if}\ m\ \mbox{is odd},\\
			&\ \frac{1}{ \tan \frac{\pi}{m}}, \ \ \quad \mbox{if}\ m\ \mbox{is even},
		\end{aligned}\right.
	\end{equation}
	and \begin{equation}\label{eq:ortho:2}
		\inf_{\norm{\vu}=\norm{\vv}=1} \sum_{j=1}^m \ucabs{\vu^\T\va_j \va_j^\T\vv} = \left\{\begin{aligned}
			&\frac{\cos \frac{\pi}{2 m}}{2\tan \frac{\pi}{2 m}}, \quad \mbox{if}\ m\ \mbox{is odd},\\
			&\ \frac{1}{ \tan \frac{\pi}{m}},             \ \ \quad \mbox{if}\ m\ \mbox{is even}.
		\end{aligned}\right.
	\end{equation}
\end{lemma}
\begin{proof}
	See Section \ref{pf:ortho:l - em:p=1}. 
\end{proof}

\section{The Condition number under \texorpdfstring{\( \ell_2 \)}{l2} norm}\label{sec:p = 2}
In this section, we first establish universal lower bounds on the condition number \(\betatwo_{\Psi_\ma}\) for all \(\ma \in \mathbb{H}^{m \times d}\). We then demonstrate the tightness or asymptotic tightness of these bounds by computing the condition numbers for the harmonic frame \(\mem \in \mathbb{R}^{m \times 2}\) and for Gaussian random matrices \(\ma \in \mathbb{H}^{m \times d}\).

\begin{theorem}\label{le:opt:number:2}
	For any matrix $\ma \in \H^{m\times d}$, one has 
	\begin{align*}
		\betatwo_{\Psi_\ma} \ge \betatwobar := \left\{\begin{aligned}
			&\sqrt{3},\quad \mbox{if } \H = \R, \\
			&\,\ 2,       \quad\; \mbox{if } \H = \C.
		\end{aligned}\right. 
	\end{align*}
\end{theorem}
\begin{proof}
	See Section \ref{pf:opt:number:2}. 
\end{proof}

It is noteworthy that, to prove Theorem \ref{le:opt:number:2}, we only need to consider the case where \( d = 2 \). In fact, letting $\mb = (\vb_1,\cdots,\vb_m)^\T\in\H^{m\times 2}$ be the the matrix composed of the first two columns of $\ma\in\hmd$, one can easily check that \begin{align*}
	L^{\ell_p}_{\Psi_\ma} = \inf_{\substack{\nj{\vu,\vv}\in\R \\ \norm{\vu}=\norm{\vv}=1}} \xkh{\sum_{j=1}^m \abs{\real{\vu^\hh\va_j \va_j^\hh\vv}}^p}^{\frac 1p} \le \inf_{\substack{\nj{\vu,\vv}\in\R, \vu,\vv\in\H^2 \\ \norm{\vu}=\norm{\vv}=1}}  \xkh{\sum_{j=1}^m \abs{\real{\vu^\hh\vb_j \vb_j^\hh\vv}}^p}^{\frac 1p}
\end{align*}
and\begin{align*}
	U^{\ell_p}_{\Psi_\ma} = \sup_{\norm{\vu}=1} \xkh{\sum_{j=1}^m \abs{\va_j^\hh\vu}^{2p}}^{\frac 1p} \ge \sup_{\vu\in\H^2, \norm{\vu}=1} \xkh{\sum_{j=1}^m \abs{\vb_j^\hh\vu}^{2p}}^{\frac 1p}.
\end{align*}
As a result, one has 
\begin{equation} \label{eq:betaAB}
	\betalp_{\Psi_\ma} = \frac{U^{\ell_p}_{\Psi_\ma}}{L^{\ell_p}_{\Psi_\ma}} \ge \frac{U^{\ell_p}_{\Psi_\mb}}{L^{\ell_p}_{\Psi_\mb}}  = \betalp_{\Psi_\mb}. 
\end{equation}

The following theorem establishes that the universal lower bound given in Theorem \ref{le:opt:number:2} for \(\mathbb{H} = \mathbb{R}\) is tight. To begin, recall the harmonic frame consisting of \(m\) equidistant points on the upper semicircle, defined as
\renewcommand{\arraystretch}{1.5}
\begin{align} \label{eq:harmf}
\mem := \begin{pmatrix}
1 & \cos\left(\frac{\pi}{m}\right) & \cdots & \cos\left(\frac{(m-1)\pi}{m}\right) \\
0 & \sin\left(\frac{\pi}{m}\right) & \cdots & \sin\left(\frac{(m-1)\pi}{m}\right)
\end{pmatrix}^\top \in \mathbb{R}^{m \times 2}.
\end{align}
\renewcommand{\arraystretch}{1}

\begin{theorem}\label{le:em:l2}
	For any $m \ge 3$, let $\mem$ be defined as in \eqref{eq:harmf}. Then
	\begin{equation*}
		\betatwo_{\Psi_\mem} = \sqrt{3}. 
	\end{equation*}
\end{theorem}
\begin{proof}
	See Section \ref{pf:em:l2}. 
\end{proof}

Next, we estimate the  condition number for a standard Gaussian random matrix, which shows that the universal lower bound given in Theorem \ref{le:opt:number:2} is asymptotic optimal for the real case.

%
\begin{theorem}\label{le:concent:p=2} 
Assume that \(\ma = (\va_1, \ldots, \va_m)^\top \in \mathbb{R}^{m \times d}\) is a standard Gaussian matrix with \(\va_j \sim \mathcal{N}(\mathbf{0}, \mi_d), j=1,\ldots,m\). For any \(0 < \delta < 4\), with probability at least \(1 - c_1 \exp(-c_0 d) - O(m^{-d})\), we have
\[
\betatwobar \leq \betatwo_{\Psi_\ma} \leq \betatwobar + \delta,
\]
provided that \(m \geq C_0 \delta^{-4} d^2\), where \(c_0, C_0 > 0\) are universal constants.
\end{theorem}
\begin{proof}
	See Section \ref{pf:concent:p=2}. 
\end{proof}

\begin{remark}
In Theorem \ref{le:concent:p=2}, as \(\delta \to 0\), the condition number \(\betatwo_{\Psi_\ma}\) converges to \(\betatwobar\). This shows that \(\betatwobar\) is asymptotically optimal in the real case. Moreover, one can check that \(\betatwobar\) is also asymptotically optimal in the complex case by analyzing the condition number of a standard complex Gaussian matrix together with the law of large numbers.
\end{remark}

%

\subsection{Proof of Theorem \ref{le:opt:number:2}} \label{pf:opt:number:2} 
According to \eqref{eq:betaAB}, without loss of generality, we assume  \( d = 2 \), i.e. $\ma = (\va_1,\va_2,\cdots,\va_m)^\hh\in \H^{m\times 2}$. We write
\begin{align}\label{eq:ai}
	\va_i = t_i \begin{pmatrix}
		\cos \varphi_i \cos \alpha_i\\
		\sin \varphi_i \cos \beta_i
	\end{pmatrix} + t_i\cdot \i \cdot \begin{pmatrix}
	\cos \varphi_i \sin \alpha_i\\
	\sin \varphi_i \sin \beta_i
	\end{pmatrix}
\end{align}
where $t_i = \norm{\va_i}\ge 0$, $(\varphi_i,\alpha_i,\beta_i)\in I^\H$ and 
\begin{align}\label{eq:ih}
	I^\H := \left\{\begin{aligned}
		&[0, \pi]\times \dkh{0}\times \dkh{0},       & \mbox{if } \H = \R, \\
		&[0, \pi]\times [0, 2\pi]\times [0, 2\pi],   & \mbox{if } \H = \C.
	\end{aligned}\right. 
\end{align}
Applying Theorem \ref{th:equivalent}, one has
 \begin{align}\label{eq:ulp}
	U^{\ell_2}_{\Psi_\ma} = \sup_{\norm{\vu}=1} \xkh{\sum_{j=1}^m \abs{\va_j^\hh\vu}^{4}}^{\frac 12} 
\end{align} and
 \begin{eqnarray}\label{eq:llp}
	L^{\ell_2}_{\Psi_\ma} &= &  \inf_{\norm{\vu}=\norm{\vv}=1} \xkh{\sumim \abs{\real{\vu^\hh\va_i \va_i^\hh\vv}}^2}^{1/2}   \notag \\
	&\le & \inf_{\substack{\norm{\vu}=\norm{\vv}=1 \\ \nj{\vu,\vv}=0}} \xkh{\sumim \abs{\real{\vu^\hh\va_i \va_i^\hh\vv}}^2}^{1/2} =: M_\ma.
\end{eqnarray}
Similar to \eqref{eq:ai}, we also parameterize $\vu, \vv$ as
 \begin{align} \label{eq:uv}
	\vu = \begin{pmatrix}
		\cos \theta \cos \alpha\\
		\sin \theta \cos \beta
	\end{pmatrix} + \i \cdot \begin{pmatrix}
		\cos \theta \sin \alpha\\
		\sin \theta \sin \beta
	\end{pmatrix}\ \ \mbox{and}\ \ \vv = \begin{pmatrix}
		\sin \theta \cos \alpha\\
		-\cos \theta \cos \beta
	\end{pmatrix} + \i \cdot \begin{pmatrix}
		\sin \theta \sin \alpha\\
		-\cos \theta \sin \beta
	\end{pmatrix},
\end{align}
where $(\theta,\alpha,\beta)\in I^\H$. A straightforward calculation reveals that 
\[
\abs{\va_i^\hh \vu}^2 = \frac 12 t_i^2 + \frac 12 t_i^2\xkh{\cos 2\varphi_i \cos 2\theta + \sin 2\varphi_i \sin 2\theta \cos\xkh{\alpha-\beta-\alpha_i+\beta_i}}
\]
and 
\begin{equation} \label{eq:Reuv}
\real{\vu^\hh\va_i\va_i^\hh\vv} = \frac 12 t_i^2 \xkh{\cos 2\varphi_i \sin 2\theta - \sin 2\varphi_i \cos 2\theta \cos\xkh{\alpha-\beta-\alpha_i+\beta_i}}.
\end{equation}
To prove the results, it suffices to show 
 \begin{align} \label{eq:upplow}
	(U_{\Psi_\ma}^{\ell_2})^2 \ge \left\{\begin{aligned}
		&\frac{3}{8}\sumim t_i^4, \  \mbox{if } \H = \R, \\
		&\frac{1}{3}\sumim t_i^4, \  \mbox{if } \H = \C,
	\end{aligned}\right. \quad \mbox{and}\quad (M_\ma)^2 \le \left\{\begin{aligned}
		&\; \frac{1}{8}\ \sumim t_i^4, \  \mbox{if } \H = \R, \\
		&\frac{1}{12}\sumim t_i^4,     \  \mbox{if } \H = \C.
	\end{aligned}\right.  
\end{align}
We first estimate $(U_{\Psi_\ma}^{\ell_2})^2$.  It follows from \eqref{eq:ulp} that 
 \[(U_{\Psi_\ma}^{\ell_2})^2 = \maxm{(\vu, \vv)\in\X^\H} \sum_{i=1}^{m} \abs{\va_i^\hh\vu}^4 = \maxm{(\theta, \alpha, \beta)\in I^\H} f(\theta,\alpha,\beta) \]
where \[f(\theta,\alpha,\beta) := \frac 14 \sum_{i=1}^{m} t_i^4\xkh{1+\cos 2\varphi_i\cos 2\theta + \sin 2\varphi_i\sin 2\theta\cos(\alpha-\beta-\alpha_i+\beta_i)}^2. \]
Consider the following two cases:

{\bf Case 1: $\H = \R$.} For this case, \( \alpha = \beta = 0 \) and \( \alpha_i = \beta_i = 0 \) for $i = 1,\cdots,m$, so \( f(\theta,0,0) \) can be simplified to 
\[f(\theta,0,0) = \frac 14 \sum_{i=1}^{m} t_i^4\xkh{1 + \cos(2\varphi_i-2\theta)}^2 = \frac 18 \sum_{i=1}^{m} t_i^4\xkh{3 + 4\cos(2\theta-2\varphi_i)+\cos(4\theta-4\varphi_i)}.\]
Through a straightforward calculation, we obtain 
\begin{eqnarray*}
	\maxm{\theta\in [0,\pi]} f(\theta,0,0) \ge \frac 1\pi \int_0^\pi f(\theta,0,0) \dd \theta \overset{\text{(i)}}{=} \frac 1\pi \int_0^\pi \frac 38 \sum_{i=1}^{m} t_i^4 \dd \theta = \frac 38 \sum_{i=1}^{m} t_i^4. 
\end{eqnarray*} 
Here, (i) comes from the fact that
\begin{align*}
	& \int_0^\pi \sum_{i=1}^{m} t_i^4 \xkh{4\cos(2\theta-2\varphi_i)+\cos(4\theta-4\varphi_i)} \dd \theta = 0 
\end{align*}
since $\cos(2\theta-2\varphi_i)$ and $\cos(4\theta-4\varphi_i)$ are periodic function  in \( \theta \) with period $\pi$. 

{\bf Case 2: $\H = \C$.} Note that \(\cos(\alpha-\beta-\alpha_i+\beta_i) = \cos(\alpha-\beta)\cos(\alpha_i-\beta_i) + \sin(\alpha-\beta)\sin(\alpha_i-\beta_i)\).  Taking $x = \cos2\theta, y = \sin 2\theta \cos(\alpha-\beta), z = \sin 2\theta \sin(\alpha-\beta), a_i = \cos2\varphi_i, b_i = \sin2\varphi_i \cos(\alpha_i-\beta_i), c_i = \sin2\varphi_i \sin(\alpha_i-\beta_i)$, then  one can rewrite \( f(\theta,\alpha,\beta) \) as 
\[f(\theta,\alpha,\beta) = \frac 14\sumim t_i^4 \xkh{1+(a_i\cdot x + b_i\cdot y + c_i\cdot z)}^2:= h(x,y,z).
\]
It is easy to see that $x^2 + y^2 + z^2 = 1$ and $a_i^2 + b_i^2 + z_i^2 = 1$ for each $i$. Denote \(\SSS^2 = \dkh{(x, y, z)^\T \in \R^3:x^2 + y^2 + z^2 = 1}\).  Since $h(x,y,z)$ is nonnegative in \(\SSS^2\), we see that \begin{align*}
	\maxm{(x,y,z)\in\SSS^2} h(x,y,z) \ge \frac{1}{4\pi} \iint_{\SSS^2} h(x,y,z)\dd S = \frac{1}{16\pi} \sumim t_i^4 \iint_{\SSS^2} \xkh{1+(a_i\cdot x + b_i\cdot y + c_i\cdot z)}^2 \dd S. 
\end{align*}
The Poisson formula in surface integrals implies  that 
\[\iint_{\SSS^2} \xkh{1+(a_i\cdot x + b_i\cdot y + c_i\cdot z)}^2 \dd S = 2\pi\int_{-1}^{1}\xkh{1+u}^2\dd u = \frac{16\pi}{3}.\]
As a result, we immediately have 
\[\maxm{(\theta,\alpha,\beta)\in I^\H}f(\theta,\alpha,\beta) \ge \frac 13 \sumim t_i^4.\]
This gives the lower bounds of  $(U_{\Psi_\ma}^{\ell_2})^2$ in \eqref{eq:upplow}.

Next, we turn to estimate the upper bound of $M_\ma$ in \eqref{eq:upplow}. According to \eqref{eq:Reuv}, one has
\((M_\ma)^2 = \min_{(\theta,\alpha,\beta)\in I^\H} g(\theta, \alpha, \beta)\) where
\begin{equation*}
	g(\theta, \alpha, \beta) :=  \frac 14 \sum_{i=1}^{m} t_i^4 \xkh{\cos 2\varphi_i\sin 2\theta - \sin 2\varphi_i\cos 2\theta\cos(\alpha-\beta-\alpha_i+\beta_i)}^2.
\end{equation*}
Using the same arguments to that of \( (U_{\Psi_\ma}^{\ell_2})^2 \), we have
\[\minm{\theta\in [0,\pi]} g(\theta,0,0) = \minm{\theta\in [0,\pi]} \frac 14 \sum_{i=1}^{m} t_i^4\xkh{\sin(2\theta - 2\varphi_i)}^2\le \frac 1{4\pi} \int_0^\pi \sum_{i=1}^{m} t_i^4\xkh{\sin(2\theta - 2\varphi_i)}^2  \dd \theta = \frac 18 \sum_{i=1}^{m} t_i^4\]
for \(\H = \R\), and 
\[\minm{(\theta, \alpha, \beta)\in I^\C} g(\theta, \alpha, \beta) \le \frac{1}{16\pi} \cdot 2\pi\int_{-1}^{1} u^2\dd u \sum_{i=1}^{m} t_i^4 = \frac{1}{12} \sum_{i=1}^{m} t_i^4\]
for \(\H = \C\). This completes the proof.

\subsection{Proof of Theorem \ref{le:em:l2}}\label{pf:em:l2}
Let $\vu = (\cos \theta \; \sin \theta)^\T$ and $\vv = (\cos \varphi \; \sin \varphi)^\T$ with \(0 \le \theta, \varphi \le 2\pi\). Denote $\mem = (\va_1,\va_2,\cdots,\va_m)^\T\in \R^{m\times 2}$. Recalling the definition of $\mem$ in \eqref{eq:harmf}, we have \renewcommand{\arraystretch}{1.5} \begin{align*}
	\mem\vu = \begin{pmatrix}
		1\cdot \cos \theta + 0\cdot \sin \theta \\
		\cos \frac{\pi}{m}\cdot \cos \theta + \sin \frac{\pi}{m}\cdot \sin \theta \\
		\cdots \\
		\cos \frac{(m-1)\pi}{m}\cdot \cos \theta + \sin \frac{(m-1)\pi}{m}\cdot \sin \theta
	\end{pmatrix} = \begin{pmatrix}
		\cos \xkh{0-\theta} \\
		\cos \xkh{\frac{\pi}{m}-\theta} \\
		\cdots \\
		\cos \xkh{\frac{(m-1)\pi}{m}-\theta} 
	\end{pmatrix}.
\end{align*}
It gives
\begin{align*}
 \sumjm \ucabs{\va_j^\T \vu}^4 = \sum_{j=1}^m \cos^4 \xkh{\frac{j\pi}{m}-\theta} = \sum_{j=1}^m \zkh{\frac 38 + \frac{\cos \xkh{\frac{2j\pi}{m}-2\theta}}{2} + \frac{\cos \xkh{\frac{4j\pi}{m}-4\theta}}{8}}.
\end{align*}
Applying the Lagrange's trigonometric identities that
\begin{equation} \label{eq:Lind} 
\sum_{j=1}^k \cos(j \theta)=\frac{\sin\xkh{\frac{2k+1}{2}\theta}}{2\sin\xkh{\frac{\theta}{2}}}-\frac12 \qquad \mbox{and} \qquad \sum_{j=1}^k \sin(j \theta)=\frac{\sin\xkh{\frac{k+1}{2}\theta} \sin\xkh{\frac{k\theta}{2}}}{\sin\xkh{\frac{\theta}{2}}},
\end{equation}
one can check that 
 \begin{equation} \label{eq:cos24}
 \sum_{j=1}^m  \cos \xkh{\frac{2j\pi}{m}-2\theta}=0, \qquad \sum_{j=1}^m  \cos \xkh{\frac{4j\pi}{m}-4\theta}=0.
 \end{equation}
Therefore, the upper Lipschitz constant of $\mem$ is 
\begin{align*}
	U_\mem^{\ell_2} = \maxm{\theta\in[0,2\pi]} \xkh{\sumjm \ucabs{\va_j^\T \vu}^4}^{1/2} = \sqrt{\frac{3m}{8}}
\end{align*}
Similarly, we have
\begin{align*}
	\mem\vv = \begin{pmatrix}
		\cos \xkh{0-\varphi} \ 
		\cos \xkh{\frac{\pi}{m}-\varphi} \ 
		\cdots \ 
		\cos \xkh{\frac{(m-1)\pi}{m}-\varphi} 
	\end{pmatrix}^\T.
\end{align*}\renewcommand{\arraystretch}{1}
This gives
\begin{align*}
\sumjm \ucabs{\va_j^\T \vu}^2\ucabs{\va_j^\T \vv}^2 =& \sum_{j=1}^m \cos^2 \xkh{\frac{j\pi}{m}-\theta} \cos^2 \xkh{\frac{j\pi}{m}-\varphi} \\
	\overset{\text{(i)}}{=} & \frac m4 + \frac 14 \sum_{j=1}^m \cos \xkh{\frac{2j\pi}{m}-2\theta}\cos \xkh{\frac{2j\pi}{m}-2\varphi} \\
	=& \frac m4 + \frac 18 \sum_{j=1}^m \zkh{\cos (2\theta-2\varphi)+ \cos \xkh{\frac{4j\pi}{m}-2\theta-2\varphi}} \\
	=& \frac m4 + \frac m8 \cos (2\theta-2\varphi). 
\end{align*}
where  (i) comes from trigonometric identities together with \eqref{eq:cos24}. 
As a result, $L_\mem^{\ell_2}$ can be calculated as follows 
\begin{align*}
	L_\mem^{\ell_2} = \minm{\theta,\varphi\in[0,2\pi]} \xkh{\sumjm \ucabs{\va_j^\T \vu}^2\ucabs{\va_j^\T \vv}^2}^{1/2} = \sqrt{\frac{m}{8}},
\end{align*}
where the last equality holds if and only if $\abs{2\theta-2\varphi}=\pi$, namely, $\vu \perp \vv$. 
Combining the lower and upper Lipschitz constants, we arrive at the conclusion that 
\begin{align*}
	\betatwo_{\mem} = U_\mem^{\ell_2}/L_\mem^{\ell_2} = \sqrt{3}.
\end{align*}

\subsection{Proof of Theorem \ref{le:concent:p=2}}\label{pf:concent:p=2} 
Based on Lemma \ref{le:u:p=2}, with probability at least \(1 - 2\exp(-d/2)\), it holds
\begin{equation*}
	\norms{\ma}_{2\rightarrow 4} \le  (3m)^{1/4} + 2\sqrt{d}.
	\end{equation*}
Therefore, for any fixed $0<\delta_1 \le 1$, one has
 \begin{equation} \label{eq:gauup}
	 \frac{1}{\sqrt{m}} U_{\Psi_\ma}^{\ell_2}= \frac{1}{\sqrt{m}} \norms{\ma}_{2\rightarrow 4}^2 \le \sqrt{3} + \delta_1,
	 \end{equation}
provided that \(m \ge C_0 \delta_1^{-4} d^2\) for some  constant \( C_0 > 0 \).   For the lower Lipschitz constant, recall that 
\[
\frac 1m (L_{\Psi_\ma}^{\ell_2})^2 = \inf_{\norm{\vu} = \norm{\vv} = 1} \frac 1m \sum_{j=1}^{m} \abs{\vu^\T \va_j\va_j^\T \vv}^2. 
\]
Since $\va_1,\ldots,\va_m \in \R^d $ are i.i.d. Gaussian random vectors,  Lemma \ref{le:l:p=2} yields that 
\begin{equation*} \label{eq:dellow}
\inf_{\norm{\vu} = \norm{\vv} = 1} \frac 1m \sum_{j=1}^{m} \abs{\vu^\T \va_j\va_j^\T \vv}^2 \ge  \inf_{\norm{\vu} = \norm{\vv} = 1}  \xkh{\norm{\vu}^2 \norm{\vv}^2+ 2\abs{\nj{\vu,\vv}}^2 - \delta_1} \ge 1-\delta_1
\end{equation*}
with probability at least $1 - c' \exp\xkh{-c'' \delta_1^2 m} - O(m^{-d})$ when $m\ge C_0  \delta_1 ^{-2} d \log d$.  Here,  $c', c'' > 0$ are universal constants. 
Thus, when $0<\delta_1  \le 1$, it holds
\begin{equation} \label{eq:gaulow}
\frac 1 {\sqrt m} L_{\Psi_\ma}^{\ell_2} \ge 1 - \delta_1.
\end{equation}
Combining \eqref{eq:gauup} and \eqref{eq:gaulow}, we obtain that for any fixed $0<\delta_1<0.4$, when \(m \ge C_0 \delta_1^{-4} d^2\), it holds 
 \begin{align*}
	\frac{U_{\Psi_\ma}^{\ell_2}}{L_{\Psi_\ma}^{\ell_2}} \le  \sqrt{3}+10\delta_1 
	\end{align*}
with probability at least $1 - c_1 \exp\xkh{-c_0 d} - O(m^{-d})$. Here, $c_1, c_0>0$ are universal constants. Taking $\delta=10 \delta_1$, we arrive at the conclusion.

\section{The condition number under \texorpdfstring{\( \ell_1 \)}{l1} norm} \label{sec:p = 1}
In this section, we study the condition number of \(\Psi\) with respect to the \(\ell_1\) norm. The following theorem establishes a universal lower bound for \(\beta^{\ell_1}_{\Psi_\ma}\).

\begin{theorem}\label{le:opt:number:l1}
	For any matrix $\ma \in \H^{m\times d}$, one has 
	\begin{align*}
		\betaone_{\Psi_\ma} \ge \betaonebar := \left\{\begin{aligned}
			&\pitwo,  \   &  \mbox{if } \ \H = \R, \\
			&\, 2,    \   &  \mbox{if } \ \H = \C.
		\end{aligned}\right. 
	\end{align*}
\end{theorem}
\begin{proof}
	See Section \ref{pf:opt:number:1}. 
\end{proof}

In the real case where \(\mathbb{H} = \mathbb{R}\), a more refined analysis allows for a slight improvement of the universal lower bound presented in Theorem \ref{le:opt:number:l1}.

\begin{theorem}\label{le:subopt:number:l1}
	Let \(d \ge 2\), \(m \ge 3\) and let \(\ma = (\va_1,\va_2,\cdots,\va_m)^\T \in \R^{m\times d}\). Then one has \begin{equation*}
		\betaone_{\Psi_\ma} \ge m \tan\pimtwo. 
	\end{equation*}
\end{theorem}
\begin{proof}
	See Section \ref{pf:subopt:number:l1}. 
\end{proof}

The following theorem characterizes the condition number \(\betaone_{\mem}\) of the harmonic frame \(\mem \in \mathbb{R}^{m \times 2}\) and demonstrates that the bound \(\pitwo\) given in Lemma \ref{le:opt:number:l1} is nearly optimal.

\begin{theorem}\label{le:em:l1}
	For \(d \ge 2\), \(m \ge 3\) and the harmonic frame \( \mem \) defined in \eqref{def:mem}, we have that \begin{equation*}
		\betaone_{\mem} = \left\{\begin{aligned}
			& m \frac{\tan \pimtwo}{\cos \pimtwo},   \quad \mbox{if}\ m\ \mbox{is odd},\\
			& \frac m2 \tan \frac{\pi}{m},           \quad \mbox{if}\ m\ \mbox{is even}.
		\end{aligned}\right.
	\end{equation*}
\end{theorem}
\begin{proof}
	See Section \ref{pf:em:l1}. 
\end{proof}

As before, the following theorem demonstrates that the universal lower bounds established in Theorem \ref{le:opt:number:l1} are asymptotically tight for both the real and complex cases.

\begin{theorem}\label{le:concent:p=1} 
	Assume that $\ma=(\va_1,\cdots,\va_m)^\T\in\hmd$ is a standard Gaussian matrix with $\H=\R$ or $\C$. For any $ 0 < \delta < 0.4 $, with probability at least $1 - 4 \exp(-c_0 d)$, one has \begin{equation}
		\betaonebar \le \betaone_{\Psi_\ma} \le \betaonebar + \delta,
	\end{equation}
	provided that $m\ge C_0\log (1/\delta) \delta^{-2} d$. Here, $c_0,C_0>0$ are some universal constants. 
\end{theorem}
\begin{proof}
	See Section \ref{pf:concent:p=1}. 
\end{proof}

\subsection{Proof of Theorem \ref{le:opt:number:l1}} \label{pf:opt:number:1} 
The proof is similar to that of Theorem \ref{le:opt:number:2}, so we only give a brief description. Without loss of generality, we only need to consider the case where  \( d = 2 \), i.e. $\ma = (\va_1,\va_2,\cdots,\va_m)^\hh\in \H^{m\times 2}$.  According to Theorem \ref{th:equivalent}, one has 
\begin{equation}\label{eq:u:l1} 
	U_{\Psi_\ma}^{\ell_1} = \sup_{\norm{\vu}=1} \norms{\ma\vu}_{2}^2 = \norm{\ma}^2 \ge \frac 12 \normf{\ma}^2
\end{equation}
and
 \begin{eqnarray}\label{eq:llpl1}
	L^{\ell_1}_{\Psi_\ma} =  \inf_{\norm{\vu}=\norm{\vv}=1} \sumim \abs{\real{\vu^\hh\va_i \va_i^\hh\vv}} \le  \inf_{\substack{\norm{\vu}=\norm{\vv}=1 \\ \nj{\vu,\vv}=0}} \sumim \abs{\real{\vu^\hh\va_i \va_i^\hh\vv}}.
\end{eqnarray}
Using the same notations as in  \eqref{eq:ai} and \eqref{eq:uv}, one has $\normf{\ma}^2=\sumim t_i^2$ and 
\begin{equation}\label{eq:ax ay}
		\real{\vu^\hh\va_i\va_i^\hh\vv} = \frac 12 t_i^2 \xkh{\cos 2\varphi_i \sin 2\theta - \sin 2\varphi_i \cos 2\theta \cos\xkh{\alpha-\beta-\alpha_i+\beta_i}},
\end{equation}
where $(\theta,\alpha,\beta)\in I^\H$ and  $(\varphi_i,\alpha_i,\beta_i)\in I^\H$ for all $i=1,\ldots,m$. Here, $I^\H$ is defined in \eqref{eq:ih}. Putting \eqref{eq:ax ay} into \eqref{eq:llpl1}, we obtain that 
 \begin{equation}\label{eq:m:l1}
	L^{\ell_1}_{\Psi_\ma} \le  \minm{(\theta,\alpha,\beta)\in I^\H} f(\theta, \alpha, \beta)
\end{equation}
where \[f(\theta, \alpha, \beta) := \frac 12 \sum_{i=1}^{m} t_i^2 \abs{\cos 2\varphi_i \sin 2\theta - \sin 2\varphi_i \cos 2\theta \cos\xkh{\alpha-\beta-\alpha_i+\beta_i}}.
\]

For the real case where $\H=\R$,  \( \alpha = \beta = 0 \) and \( \alpha_i = \beta_i = 0 \) for $i = 1,\cdots,m$. Furthermore, one has
\begin{align*}
	\minm{\theta\in [0,2\pi]} f(\theta,0,0) \le \frac{1}{\pi} \int_{0}^{\pi} f(\theta,0,0)\dd\theta = \sumim \frac{t_i^2}{2\pi} \int_{\varphi_i}^{\varphi_i + \pi} \abs{\sin(2\theta - 2\varphi_i)} \dd\theta \\
	\overset{\text(i)}{=} \sumim \frac{t_i^2}{\pi} \int_{\varphi_i}^{\varphi_i + \pitwo} \sin(2\theta - 2\varphi_i) \dd\theta = \frac{1}{\pi} \sum_{i=1}^{m} t_i^2,
\end{align*}
where (i) comes from the fact that  \(\abs{\sin(2\theta - 2\varphi_i)}\) is a periodic function in \(\theta\) with period \(\pitwo\). 

For the complex case where $\H=\C$, the  Poisson formula in surface integrals reveals that 
\[\minm{(\theta, \alpha, \beta)\in I^\C} f(\theta, \alpha, \beta) \le \frac{1}{8\pi} \cdot 2\pi\int_{-1}^{1} \abs{u}\dd u \sum_{i=1}^{m} t_i^2 = \frac{1}{4} \sum_{i=1}^{m} t_i^2.
\]
Putting everything together, we obtain the conclusion 
\begin{align*}
	\betaone_{\Psi_\ma} \ge \frac{U^{\ell_1}_\ma}{L^{\ell_1}_{\Psi_\ma}} \ge \betaonebar.
\end{align*}

\subsection{Proof of Theorem \ref{le:subopt:number:l1}} \label{pf:subopt:number:l1}
Our analysis primarily relies on the following lemma. While \cite[Lemma 3.6]{xiayu2024stability} establishes a lower bound for \(\max_{\theta \in [0,\pi]} \sum_{i=1}^m t_i^2 \left| \sin(\theta - \varphi_i) \right|\), this paper requires an upper bound for \(\min_{\theta \in [0,\pi]} \sum_{i=1}^m t_i^2 \left| \sin(\theta - \varphi_i) \right|\). It is worth noting that the proof of the next lemma is considerably more involved.

\begin{lemma}\label{lemma:sub:tan}
	For any $\varphi_1, \varphi_2, \cdots, \varphi_{m+1} \in \R$ satisfying $0 \le \varphi_1 \le \varphi_2 \le \cdots \le \varphi_m \le \varphi_{m+1} = \pi$, and for all $t_1, t_2, \cdots, t_{m} \in \R$, we have
	 \begin{equation}
		\minm{\theta \in [0,\pi]} \sum_{i = 1}^m t_i^2\abs{\sin\xkh{\theta-\varphi_i}} \le \frac{1}{m\cdot \tan \frac{\pi}{2 m}} \sum_{i = 1}^m t_i^2.
	\end{equation}
\end{lemma}

With Lemma \ref{lemma:sub:tan} in place, we can give a proof of Theorem \ref{le:subopt:number:l1}.

\begin{proof}[Proof of Theorem \ref{le:subopt:number:l1}]
According to \eqref{eq:betaAB}, we only need to consider the case where  \( d = 2 \).  
Similar to \eqref{eq:u:l1} and \eqref{eq:llpl1}, we may assume $\va_i=t_i(\cos\varphi_i,\sin\varphi_i)^\T$ with $t_i=\norm{\va_i}$ and $0\le \varphi_i \le \pi$ for all $i=1,\ldots,m$, and then have that 
\[
U_{\Psi_\ma}^{\ell_1}  \ge \frac12 \sum_{i = 1}^m t_i^2
\]
and 
\[
L^{\ell_1}_{\Psi_\ma}  \le  \inf_{\substack{\norm{\vu}=\norm{\vv}=1 \\ \nj{\vu,\vv}=0}} \sumim \abs{\vu^\hh\va_i \va_i^\hh\vv}.
\]
Denote $\vu=(\cos\theta,\sin\theta)^\T$ and $\vv=(\sin\theta,-\cos\theta)^\T$ with $0\le \theta\le \pi$. Then 
\begin{eqnarray*}
\inf_{\substack{\norm{\vu}=\norm{\vv}=1 \\ \nj{\vu,\vv}=0}} \sumim \abs{\vu^\hh\va_i \va_i^\hh\vv} =  \minm{\theta\in [0,\pi]} \frac 12 \sum_{i=1}^{m} t_i^2 \abs{\sin(2\theta-2\varphi_i)} =  \minm{\theta\in [0,\pitwo]} \frac 12 \sum_{i=1}^{m} t_i^2 \abs{\sin(2\theta-2\varphi_i)} ,  
\end{eqnarray*}
where the last equality comes from the fact that $\abs{\sin(2\theta-2\varphi_i)}$ has a period of $\pi/2$ for $\theta$ and $\varphi_i$. 
Therefore, apply Lemma \ref{lemma:sub:tan} to obtain that   
\begin{eqnarray*}
	\inf_{\substack{\norm{\vu}=\norm{\vv}=1 \\ \nj{\vu,\vv}=0}} \sumim \abs{\vu^\hh\va_i \va_i^\hh\vv} \le  \frac{1}{2 m\cdot \tan \frac{\pi}{2 m}} \sum_{i = 1}^m t_i^2.  
\end{eqnarray*}  
Combining the above two bounds, we conclude
 \[\betaone_{\Psi_\ma} \ge \frac{U_\ma^{\ell_1}}{L^{\ell_1}_{\Psi_\ma} }  \ge m\tan \frac{\pi}{2 m}.\]
\end{proof}

Finally, we give the proof of Lemma \ref{lemma:sub:tan}.

\begin{proof}[Proof of Lemma \ref{lemma:sub:tan}]
	For the sake of simplicity, we take $g(\theta) := \sumim t_i^2 \abs{\sin\xkh{\theta-\varphi_i}}$.  For each $1\le k\le m$, when $\theta\in\zkh{\varphi_k,\varphi_{k+1}}$, it holds
	\[
	g(\theta) =\sum_i^k t_i^2 \sin\xkh{\theta-\varphi_i} - \sum_{i = k+1}^m t_i^2 \sin\xkh{\theta-\varphi_i} :=g_k(\theta) .
	\]
	For each $1\le k\le m$, there exists $\theta_k\in\zkh{0,\pi}$ and $\rk \ge 0$ such that 
		 \begin{equation}
		g_k(\theta)  = r_k \sin\xkh{\theta-\theta_k}.
	\end{equation}
	Therefore, one has 
	\begin{align*}
		\int_{\phik}^{\phikk} g(\theta) \dd \theta &= \int_{\phik}^{\phikk} \rk\sin\xkh{\theta - \thetak} \dd \theta \\
		&= \rk\cdot\xkh{\cos\xkh{\phik-\thetak}-\cos\xkh{\phikk-\thetak}}\\
		&= 2\cdot \rk\cdot \sin\xkh{\frac{\phik+\phikk}{2}-\thetak} \sin\xkh{\frac{\phikk-\phik}{2}} \\
		&= 2 \gk\xkh{\frac{\phik+\phikk}{2}} \sin\xkh{\frac{\phikk-\phik}{2}} 
	\end{align*}
	due to $\phik \le \frac{\phik+\phikk}{2} \le \phikk$. We claim that the following holds: 
	\begin{equation}\label{claim:g middle}
		\gk\xkh{\frac{\phik+\phikk}{2}} \cos\xkh{\frac{\phikk-\phik}{2}} \ge \minm{\theta\in\zeropi} g(\theta)
	\end{equation}
	for each $k = 1,\cdots,m$.  We next divide the proof into two cases. 
	\begin{itemize}
		\item For all \(1 \le k \le m\), it holds \(\frac{\varphi_{k+1}-\varphi_k}{2} < \frac{\pi}{2}\). Then we have 
		\begin{align*}
			\sumim t_i^2 &= \frac 12 \int_0^\pi g(\theta) \dd\theta = \frac 12 \sumkm \int_{\phik}^{\phikk} g(\theta) \dd\theta \\
			&= \sumkm \gk\xkh{\frac{\phik+\phikk}{2}} \cos\xkh{\frac{\phikk-\phik}{2}} \tan\xkh{\frac{\phikk-\phik}{2}} \\
			&\ge \minm{\theta\in\zeropi} g(\theta) \cdot \sumkm \tan\xkh{\frac{\phikk-\phik}{2}} \\
			&\stackrel{\mbox{(i)}}{\ge} \minm{\theta\in\zeropi} g(\theta) \cdot m \cdot \tan\xkh{\sumim\frac{\phikk-\phik}{2m}} = \minm{\theta\in\zeropi} g(\theta) \cdot m \tan \frac{\pi}{2m} ,
		\end{align*}
		where (i) comes from Jensen's inequality due to the fact that \(\tan x\) is a convex function on \([0, \frac{\pi}{2})\).  We arrive at the conclusion for this case.
		
		\item There exists a \(k \) such that \(\frac{\varphi_{k+1}-\varphi_k}{2} = \frac{\pi}{2}\), i.e., \(\varphi_{k+1}-\varphi_k = \pi\).  For this case, we have \(\varphi_j = 0\) for all $j \le k$ and \(\varphi_j = \pi \)  for all $j \ge k+1$.
Thus, one can conclude that \begin{equation*}
			\minm{\theta \in [0,\pi]} g(\theta) = \minm{\theta \in [0,\pi]} \sum_{i = 1}^m t_i^2 \abs{\sin\xkh{\theta}} = 0.
		\end{equation*}
This implies the conclusion of Lemma \ref{lemma:sub:tan} holds trivially  for \(m\ge 2\).
	\end{itemize}

	Finally, it remains to prove the claim \eqref{claim:g middle}. For any fixed $k\in\dkh{1,\cdots,m}$,  note that $g(\theta) = \gk (\theta) = \rk \sin\xkh{\theta-\thetak} \ne 0$ when $\theta\in\xkh{\phik,\phikk}$. Therefore, \(g\xkh{\theta}\) is either monotonic on the interval $\zkh{\phik,\phikk}$, or it first increases and then decreases on $\zkh{\phik,\phikk}$. For convenience, we define $\alpha_k=\phik-\thetak$ and $\alpha_{k+1}=\phikk-\thetak$.  Observe that when $\theta\in\zkh{\varphi_k,\varphi_{k+1}}$, it holds $g_k(\theta)=\rk\sin\xkh{\theta-\thetak}=g(\theta) \ge 0$. It means that 
 $0\le \alphak\le\alphakk\le\pi$.  Our proof then divides into the following three cases:
	
	{\bf Case 1: $g(\theta)$ is monotonically increasing on $\zkh{\phik,\phikk}$.} For this case, one has $0\le \alphak\le\alphakk\le\frac{\pi}{2}$. Therefore,  
	 \begin{align*}
		\minm{\theta\in\zeropi} g(\theta) &\le \minm{\theta\in\zkh{\phik,\phikk}} g(\theta) = \rk \sin\xkh{\alphak} 
		= \rk \sin\xkh{\frac{\alphakk+\alphak}{2} - \frac{\alphakk-\alphak}{2}} \\
		&= \rk \xkh{\sin \frac{\alphakk+\alphak}{2} \cos \frac{\alphakk-\alphak}{2} - \cos \frac{\alphakk+\alphak}{2} \sin \frac{\alphakk-\alphak}{2}}\\
		&\le \rk \sin \frac{\alphakk+\alphak}{2} \cos \frac{\alphakk-\alphak}{2} = \gk \xkh{\frac{\alphakk+\alphak}{2}} \cos \frac{\alphakk-\alphak}{2}
	\end{align*}
	where the last inequality holds due to $\rk\ge 0$ and the fact that 
	\begin{equation*}\label{ineq:case 1}
		\cos\frac{\alphakk+\alphak}{2} \sin\frac{\alphakk+\alphak}{2}\ge 0.
	\end{equation*}

	{\bf Case 2: $g(\theta)$ is monotonically decreasing on $\zkh{\phik,\phikk}$.} This implies that $\frac{\pi}{2}\le \alphak\le\alphakk\le\pi$ and \begin{equation*}
		\cos\frac{\alphakk+\alphak}{2} \sin\frac{\alphakk+\alphak}{2}\le 0.
	\end{equation*}
	Therefore, we obtain that \begin{align*}
		\minm{\theta\in\zeropi} g(\theta) &\le \minm{\theta\in\zkh{\phik,\phikk}} g(\theta) = \rk \sin\xkh{\alphakk} 
		= \rk \sin\xkh{\frac{\alphakk+\alphak}{2} + \frac{\alphakk-\alphak}{2}} \\
		&= \rk \xkh{\sin \frac{\alphakk+\alphak}{2} \cos \frac{\alphakk-\alphak}{2} + \cos \frac{\alphakk+\alphak}{2} \sin \frac{\alphakk-\alphak}{2}}\\
		&\le \rk \sin \frac{\alphakk+\alphak}{2} \cos \frac{\alphakk-\alphak}{2} = \gk \xkh{\frac{\alphakk+\alphak}{2}} \cos \frac{\alphakk-\alphak}{2}.
	\end{align*}
	
	{\bf Case 3: $g(\theta)$ is increasing first and then decreasing on $\zkh{\phik,\phikk}$.} If the minimum value on $\zkh{\phik,\phikk}$ is attained at \(\phik\), one can see that  \(\frac{\pi}{2}-\alphak \ge \alphakk-\frac{\pi}{2}\); it implies that $\cos\frac{\alphakk+\alphak}{2} \sin\frac{\alphakk+\alphak}{2}\ge 0$, and the remaining proof details are similar to those of Case 1. Otherwise,  if the minimum value on this interval is attained at \(\phikk\), it indicates \(\frac{\pi}{2}-\alphak \le \alphakk-\frac{\pi}{2}\) and $\cos\frac{\alphakk+\alphak}{2} \sin\frac{\alphakk+\alphak}{2}\le 0$, and the remaining proof is then similar to Case 2.
	
	This completes the proof. 
\end{proof}

\subsection{Proof of Theorem \ref{le:em:l1}} \label{pf:em:l1}
For the harmonic frame \( \mem \) defined in \eqref{def:mem}, one can check that  \(\mem^\T\mem = \frac m2 \mi\). Denote \(\mem = (\va_1, \va_2, \cdots, \va_m)^T \in \R^{m\times 2}\). Then the upper Lipschitz constant 
\[U_\mem^{\ell_1} = \sup_{\norm{\vu}=1} \norms{\ma\vu}_{2}^2 = \norm{\mem} = \frac m2. \] 
On the other hand, Lemma \ref{le:ortho:l - em:p=1} reveals that \begin{equation*}
	L_\mem^{\ell_1} = \inf_{\norm{\vu}=\norm{\vv}=1} \sum_{j=1}^m \ucabs{\vu^\T\va_j \va_j^\T\vv} = \left\{\begin{aligned}
		&\frac{\cos \frac{\pi}{2 m}}{2\tan \frac{\pi}{2 m}},   \quad \mbox{if}\ m\ \mbox{is odd},\\
		&\ \frac{1}{ \tan \frac{\pi}{m}},                  \ \ \quad \mbox{if}\ m\ \mbox{is even}.
	\end{aligned}\right.
\end{equation*}
As a result, we obtain the conclusion  that 
\begin{equation*}
	\betaone_{\mem} = \frac{U_\mem^{\ell_1}}{L_\mem^{\ell_1}} = \left\{\begin{aligned}
		& m \frac{\tan \pimtwo}{\cos \pimtwo},   \quad \mbox{if}\ m\ \mbox{is odd},\\
		& \frac m2 \tan \frac{\pi}{m},           \quad \mbox{if}\ m\ \mbox{is even}.
	\end{aligned}\right.
\end{equation*}

\subsection{Proof of Theorem \ref{le:concent:p=1}} \label{pf:concent:p=1}
Given the extensive research on the concentration properties of Gaussian matrices, we directly present several essential lemmas.

\begin{lemma} \cite{candes2013phaselift, fourcart2013} \label{le:u:p=1}
	Assume that $\ma\in\H^d$ is a standard Gaussian matrix with $\H=\R$ or $\C$. Then for any $ 0 < \delta < 1 $, with probability at least $1-2\exp(-c_2 d)$, one has \begin{equation}
		(1-\delta)\norm{\vu}^2 \le \frac 1m \norm{\ma\vu}^2 \le (1+\delta)\norm{\vu}^2
	\end{equation}
	simultaneously for all $\vu\in\H^d$ provided that $m\gtrsim\delta^{-2} d$. Here, $c_2 > 0$ is a universal constant. 
\end{lemma}

\begin{lemma}\label{le:l-concent:p=1}
	Assume that $\ma\in\H^d$ is a standard Gaussian matrix with $\H=\R$ or $\C$, and $m\gtrsim\log(1/\delta) \delta^{-2} d$. Then for any $ 0 < \delta < 0.5 $, with probability exceeding $1-4\exp(-c_3 d)$ for some universal constant $c_3>0$, one has \begin{equation}
		\frac 1m \sum_{j=1}^{m} \abs{\real{\vu^\hh\va_j\va_j^\hh \vv}} \ge \E\abs{\real{\vu^\hh\va_1\va_1^\hh \vv}} - \delta 
	\end{equation}
 for all $\vu,\vv\in\SSS_\H^{d-1}$. 
\end{lemma}
\begin{proof}
The proof is similar to that in \cite[Appendix B.1]{xiayu2024stability} and is therefore omitted.
\end{proof}

With the above two lemmas in hand, we are now ready to prove Lemma \ref{le:concent:p=1}.

\begin{proof}[Proof of Lemma \ref{le:concent:p=1}]
According to Lemma \ref{le:u:p=1}, for any $ 0 < \delta < 1$, it holds with probability at least $1-2\exp(-c_2 d)$ that 
 \[\frac 1m U_{\Psi_\ma}^{\ell_1} = \sup_{\norm{\vu}=1} \frac1m\norms{\ma\vu}_{2}^2 \le 1 + \delta,
 \]
 provided $m\ge C_0\log(1/\delta) \delta^{-2} d$ for some  universal constant  $C_0>0$.  Lemma \ref{le:l-concent:p=1} reveals that when $m\ge C_0\log(1/\delta) \delta^{-2} d$, with probability exceeding $1-4\exp(-c_0 d)$, 
\[\frac 1m L_\ma^{\ell_1} = \inf_{\norm{\vu}=\norm{\vv}=1} \frac 1m \sum_{j=1}^{m} \abs{\real{\vu^\hh\va_j\va_j^\hh \vv}} \ge \inf_{\norm{\vu}=\norm{\vv}=1} \E\abs{\real{\vu^\hh\va_1\va_1^\hh \vv}} - \delta.
\]
Here, $c_3>0$ is a universal constant. Combining the two bounds together with Lemma \ref{le:l-e:p=1}, we obtain that for any $0<\delta<0.2$,  with probability exceeding $1-6\exp(-c_0 d)$, it holds
 \begin{align*}
	\betaone_{\Psi_\ma}=\frac{U_\ma^{\ell_1}}{L_\ma^{\ell_1}} \le \left\{\begin{aligned}
		&\frac \pi2+5\delta,   \quad \mbox{if }\ \H = \R, \\
		& \hspace{0.15em}2\, + 5\delta,     \quad \mbox{if }\ \H = \C,
	\end{aligned}\right. 
\end{align*}
where the inequality follows from the fact that $\frac{1+\delta}{t-\delta}\le\frac 1t +10\delta$ for any $t > 0.5$ and $0<\delta<0.2$, and  $c_0>0$ is a  universal constant. Finally, noting that the lower Lipchitz bound is given in Lemma \ref{le:opt:number:l1}, we complete the proof.
\end{proof}



\section{Discussion}\label{sec:discussion}

This article provides a systematic analysis of the stability of phase retrieval by examining the condition number of the intensity measurement map \(\Psi_{\ma}\). We establish universal lower bounds on the condition number under both the \(\ell_1\)- and \(\ell_2\)-norms. Furthermore, these bounds are shown to be tight or asymptotically tight by evaluating the condition numbers of the harmonic frame \(\mem \in \mathbb{R}^{m \times 2}\) and Gaussian random matrices \(\ma \in \mathbb{H}^{m \times d}\). We conclude by discussing potential extensions and related open problems.

\paragraph{\textit{The Condition Number for Structured Measurement Matrices.}}

In many practical applications, the sensing matrix in phase retrieval exhibits a Fourier structure. Investigating the condition number for Fourier-type phase retrieval problems-such as short-time Fourier transform (STFT) phase retrieval and coded diffraction patterns (CDP)-is of significant practical interest and warrants further study.

\paragraph{\textit{The Optimal Sensing Matrix.}}

While the harmonic frame is known to be an optimal sensing matrix minimizing the condition number in \(\mathbb{R}^{m \times 2}\) for \(p=2\), designing sensing matrices \(\ma \in \mathbb{H}^{m \times d}\) with optimal condition number in general spaces remains an open and challenging problem.

\section{Appendix A: Technical Lemmas}

%

\begin{lemma}\label{le:l-e:p=1}
	For a standard Gaussian random vector \(\va\in\H^d\) defined as \eqref{def:gauss}, one has 
	\begin{align*}
		\inf_{\substack{\vu, \vv\in\SSS_{\H}^{d-1} \\ \nj{\vu,\vv}\in\R}}\E\abs{\real{\vu^\hh\va\va^\hh\vv}} = \left\{\begin{aligned}
			&\frac 2\pi,  & \mbox{if }\; \H = \R, \\
			&\frac 12,    & \mbox{if }\; \H = \C.
		\end{aligned}\right. 
	\end{align*}
\end{lemma}
\begin{proof}
Note that $\E\abs{\real{\vu^\hh\va\va^\hh\vv}}$ is unchanged if we multiply $\vu$ or $\vv$ by $-1$. Therefore, without loss of generality we can assume that \(\nj{\vu,\vv}\ge 0\).
	Due to the rotational invariance of the Gaussian random vector, we assume that  \(\vx=\zkh{1,0,\cdots,0}^\T\) and \(\vx=\zkh{\cos\theta,\sin\theta,0,\cdots,0}^\T\), where \(\cos\theta=\nj{\vu,\vv}\) and \(\theta\in[0,\pi/2]\). Next, we divide the proof into two cases.
	
	{\bf Case 1: $\H = \R$.} A direct calculation gives (see, e.g. \cite[Lemma 4.7]{xiayu2024stability})
\[
\E\abs{\vu^\hh\va\va^\hh\vv}  =\frac 2\pi(\sin\theta+(\frac \pi2-\theta)\cos\theta):=p(\theta).
\]
One can check that $p(\theta)$ is a decreasing function over $\theta\in[0,\pi/2]$. Thus, 
\[
\inf_{ \vu, \vv \in\SSS_{\R}^{d-1}} \E\abs{\vu^\hh\va\va^\hh\vv} = \frac 2\pi .
\]
	
	{\bf Case 2: $\H = \C$.} Denote the first two entries of $\va$ as $\frac a\sqrttwo$ and $\frac b\sqrttwo$. Here, $a = a_1+a_2\i$ and $b = b_1+b_2\i$, so one has $a_1,a_2,b_1,b_2\sim\mathcal{N} (0,1)$. Then, we have
	\begin{align*}
		h(\theta)&:=\E\abs{\real{\vu^\hh\va\va^\hh\vv}} = \frac 12\E \abs{\real{a(\bar{a}\cos \theta+\bar{b}\sin\theta)}} \\
		&= \frac 12\E\xkh{\abs{\xkh{a_1^2 + a_2^2}\cos\theta + \xkh{a_1b_1 + a_2b_2}\sin\theta}} \\
		&= \frac 12\int_{-\infty}^{\infty}\int_{-\infty}^{\infty}\int_{-\infty}^{\infty}\int_{-\infty}^{\infty} \frac{1}{4\pi^2} f_1(a_1,a_2,b_1,b_2,\theta) \e^{-(a_1^2+a_2^2+b_1^2+b_2^2)/2}\dd a_1\dd a_2\dd b_1\dd b_2
	\end{align*}
	where \[f_1(a_1,a_2,b_1,b_2,\theta) = \abs{\xkh{a_1^2 + a_2^2}\cos\theta + \xkh{a_1b_1 + a_2b_2}\sin\theta}.\]
	We claim that
	 \begin{equation} \label{cla:htheta}
	h(\theta) = \frac{3+\cos2\theta}4.
	\end{equation}
	The conclusion then holds immediately by noting that 
	 \[h(\theta) \ge h(\frac \pi 2) = \frac 12.\]
It remains to prove the claim \eqref{cla:htheta}. Taking the polar coordinates transformations that $a_1 = \rho\cos\phi_1\cos\phi_2, a_2 = \rho\cos\phi_1\sin\phi_2, b_1 = \rho\sin\phi_1\cos\phi_3$ and $b_2 = \rho\sin\phi_1\sin\phi_3$, where $\rho\in [0,+\infty)$, $\phi_1\in[0,\frac \pi2)$, $\phi_2,\phi_3\in[0,2\pi)$ and $\dd a_1\dd a_2\dd b_1\dd b_2 = \rho^3\cos\phi_1\sin\phi_1 \cdot \dd \rho\dd \phi_1\dd \phi_2\dd \phi_3$,  then $h(\theta)$ can be written as 
	\begin{align*}
		h(\theta) &= \frac{1}{8\pi^2} \int_{0}^{\infty}\int_{0}^{\frac \pi2}\int_{0}^{2\pi}\int_{0}^{2\pi} \rho^5\e^{-\frac{\rho^2}{2}} \cos \phi_1\sin\phi_1 f_2(\phi_1,\phi_2,\phi_3,\theta) \cdot \dd \rho\dd \phi_1\dd \phi_2\dd \phi_3 \\
		&= \frac{1}{\pi^2} \int_{0}^{\frac \pi2}\int_{0}^{2\pi}\int_{0}^{2\pi} \cos \phi_1\sin\phi_1 f_2(\phi_1,\phi_2,\phi_3,\theta) \cdot \dd \phi_1\dd \phi_2\dd \phi_3 \\
		&=  \frac{1}{\pi} \int_{0}^{\frac \pi2}\int_{0}^{2\pi} \cos \phi_1\sin\phi_1 \abs{\cos\theta + \cos 2\phi_1 \cos\theta + \sin 2\phi_1\cos \phi_2 \sin\theta} \cdot \dd \phi_1\dd \phi_2,
	\end{align*}
	where \begin{align*}
		f_2(\phi_1,\phi_2,\phi_3,\theta) = \frac12 \abs{\cos\theta + \cos 2\phi_1 \cos\theta + \sin 2\phi_1\cos(\phi_2-\phi_3)\sin\theta},
	\end{align*}
	and the last equality comes from the fact that $f_2(\phi_1,\phi_2,\phi_3,\theta)$ is a periodic function with a period of $2\pi$ for $\phi_3$. Setting $x = \cos2\phi_1, y = \sin 2\phi_1\cos\phi_2, z = \sin 2\phi_1\sin \phi_2$ and denoting 
	$\vx = \xkh{x,y,z}^\T \in\SSS^2$, $\ve = \xkh{\cos\theta,\sin\theta,0}^\T \in\SSS^2$ with $\SSS^2 = \dkh{(x, y, z)^\T \in \R^3:x^2 + y^2 + z^2 = 1}$, one has
	\begin{align} 
		h(\theta) &= \frac{1}{4\pi} \iint_{\SSS^2}\abs{\cos\theta + x\cos\theta + y\sin\theta}\dd S= \frac{1}{4\pi} \iint_{\SSS^2}\abs{\cos\theta + \nj{\vx,\ve}}\dd S \notag \\
		&= \frac{1}{4\pi} \iint_{\SSS^2}\abs{\cos\theta + \nj{\mpm\vx,\ve}}\dd S= \frac{1}{4\pi} \iint_{\SSS^2}\abs{\cos\theta + x}\dd S, \label{eq:hds}
	\end{align}
	where $\mpm$ is a  rotation matrix defined as
	 \begin{align*}
		\mpm = \begin{bmatrix}
			\cos\theta & \sin\theta &0\\
			-\sin\theta & \cos\theta &0\\
			0          &0           &1
		\end{bmatrix}.
	\end{align*}
	Finally, to compute the surface integral in \eqref{eq:hds}, we take  another  polar transformation that by setting  $x = \cos \phi_4, y = \sin \phi_4\cos\phi_5, z = \sin \phi_4\sin \phi_5$ for \(\phi_4\in [0, \pi), \phi_5\in [0, 2\pi)\), which reveals that \begin{align*}
		h(\theta) = \frac{1}{4\pi} \int_{0}^{2\pi}\dd \phi_5 \int_{0}^{\pi} \sin\phi_4 \abs{\cos\phi_4 + \cos\theta} \cdot \dd \phi_4 = \frac{3+\cos 2\theta}{4}.
	\end{align*}
This completes the proof. 
\end{proof}

%
%
%
%
%
%
%
\begin{lemma}\label{le:u:p=2} \cite[Lemma A.5]{cai2016optimal}
	Let \(\ma = (\va_1,\va_2,\cdots,\va_m)^\T \in \R^{m\times d}\) be a standard Gaussian random matrix with  $\va_j \sim \mathcal{N}(0, \mi_d)$. Then, with probability exceeding \(1 - 2\exp(-t^2/2)\), it holds  
	\[\norms{\ma}_{2\rightarrow 4} \le (3m)^{1/4} + \sqrt{d} + t.\]
\end{lemma}

\begin{lemma}\label{le:l:p=2} 
	Suppose that $\va_1,\ldots,\va_m \in \R^d $ are i.i.d. Gaussian random vectors with $\va_j \sim \mathcal{N}(0, \mi_d)$. For any constant $0<\delta<1$, with probability at least $1 - c_1\exp\xkh{-c_0  \delta^2  m} - O(m^{-d})$, it holds 
	\begin{align}
		\frac1m \sum_{j=1}^m \aabs{\va_j^\T \vu}^2 \aabs{ \va_j^\T \vv}^2  \ge \norm{\vu}^2 \norm{\vv}^2+2\big| \vu^\T \vv\big|^2  -\delta
	\end{align}
	simultaneously for all $\vu, \vv \in \H^d$, provided that $m\ge C_0  \delta ^{-2} d \log d$. Here,  $c_0, c_1, C_0>0$ are some universal constants. 
\end{lemma}
\begin{proof}
The proof follows a similar step-by-step approach as in \cite[Lemma 21]{sunju2018} and is therefore omitted here.
\end{proof}


\section{Appendix B: Proof of Lemma \ref{le:ortho:l - em:p=1}}\label{pf:ortho:l - em:p=1}
Setting $\vx = (\cos \theta \ \sin \theta)^\T$ and $\vy = (\cos \psi \ \sin \psi)^\T$ with $0 \le \theta, \psi \le 2\pi$, a little algebra reveals
\begin{align*}
	\mem\vx =& \begin{pmatrix}
		\cos \xkh{0-\theta} \ 
		\cos \xkh{\frac{\pi}{m}-\theta} \ 
		\cdots \ 
		\cos \xkh{\frac{(m-1)\pi}{m}-\theta} 
	\end{pmatrix}^\T, \\
	\mem\vy =& \begin{pmatrix}
		\cos \xkh{0-\psi} \ 
		\cos \xkh{\frac{\pi}{m}-\psi} \ 
		\cdots \ 
		\cos \xkh{\frac{(m-1)\pi}{m}-\psi} 
	\end{pmatrix}^\T, 
\end{align*}
for $\mem = (\va_1,\va_2,\cdots,\va_m)^\T$ defined in \eqref{def:mem}. It gives  
\[\sum_{j=1}^m \ucabs{\vy^\T\va_j \va_j^\T\vx} = \sum_{j=1}^{m} \ucabs{\cos(\frac{j\pi}{m} - \theta) \cos(\frac{j\pi}{m} - \psi)} =: G(\theta,\psi).\]
Note that \(G(\theta,\psi) = G(\psi,\theta)\), \(G(\theta, \psi) = G(\theta-\pi, \psi) = G(\theta+\pi, \psi)\), and \(G(\theta, \psi) = G(\theta+\frac{\pi}{m}, \psi+\frac{\pi}{m})\). 
Without loss of generality, one assumes \(0 \le \theta \le \psi \le \pi\). 
We next claim that \begin{equation}\label{claim:em:1} 
	\minm{0\le \theta, \psi \le \pi} G(\theta, \psi)  =  \minm{\substack{\theta\in[0, \pim] \\ \phi \in [\theta, \theta+\pim]}} G(\theta, \phi+\pitwo). 
\end{equation}
To calculate \(G(\theta, \phi+\pitwo)\), we divide the interval \([\pitwo, \pi]\) into segments of length \(\pim\) and separately discuss the scenario \(\psi\) within each subinterval.  
For \(\theta\in[0, \pim]\) and \(m\ge 3\), take \(\phi \in [0, \pim]\) and denote \begin{equation*}
	G_k(\theta, \phi) := G(\theta, \phi+\pitwo+k\pim) = \sum_{j=1}^{m} \ucabs{\cos(\frac{j\pi}{m} - \theta) \sin(\frac{j\pi}{m} - \phi - k\pim)}
\end{equation*}
where \(k\) is an integer with \(0\le k \le \hat{k}\) and \begin{equation}\label{eq:khat}
	\hat{k} = \left\{\begin{aligned}
		\frac{m-2}{2},    &\quad \mbox{if}\ m\ \mbox{is even},\\ 
		\frac{m-1}{2},    &\quad \mbox{if}\ m\ \mbox{is odd and } \phi \in [0, \pimtwo],\\ 
		\frac{m-3}{2},    &\quad \mbox{if}\ m\ \mbox{is odd and } \phi \in [\pimtwo, \pim].\\ 
	\end{aligned}\right.
\end{equation}
Then one claims that \begin{equation}\label{eq:gl}
	G_k(\theta, \phi) = \left\{\begin{aligned}
		&\frac{\cos\kpim \cos(\pim - \theta - \phi)}{\sin\pim} + k\sin(\phi-\theta+\kpim) \hspace{6em}  \mbox{ if}\ m\ \mbox{is even},\\ 
		&\frac{\cos(\pimtwo + \kpim) \cos(\frac{\pi}{2m} - \theta - \phi)}{\sin\pim} + \frac{2k+1}{2}\sin(\phi-\theta+\kpim) \mbox{ if}\ m\ \mbox{is odd and } \theta \in [0, \pimtwo],\\ 
		&\frac{\cos(\pimtwo - \kpim) \cos(\frac{3\pi}{2m} - \theta - \phi)}{\sin\pim} + \frac{2k-1}{2}\sin(\phi-\theta+\kpim) \mbox{ if}\ m\ \mbox{is odd and } \theta \in [\pimtwo, \pim].\\ 
	\end{aligned}\right.
\end{equation}

Below, we provide a proof for the conclusion (b) of Lemma \ref{le:ortho:l - em:p=1}. 
To establish the minimum of  \(G(\theta, \phi+\pitwo)\) for \(\theta\in[0, \pim]\) and $\phi \in [\theta, \theta+\pim]$, we divide \(\phi\in[\theta, \theta+\pim]\) into the two cases \([\theta, \pim]\) and \([\pim, \theta+\pim]\) for discussion. 
\begin{itemize}
	\item When \(\phi\in[\theta, \pim]\), we consider \(G_0(\theta, \phi)\). 
	If $m$ is even, \begin{equation*}
		G_0(\theta, \phi) = \frac{\cos(\pim - \theta - \phi)}{\sin\pim} \ge \frac{\cos\pim}{\sin\pim}
	\end{equation*}
	where the inequality holds if and only if \(\theta + \phi = 0\) or \(\frac{2\pi}{m}\), i.e., \(\theta = \phi = 0\) or \(\theta = \phi = \pim\). If $m$ is odd, one has \begin{equation*}
		G_0(\theta, \phi) = \left\{\begin{aligned}
			&\frac{\cos\theta \cos(\frac{\pi}{2m} - \phi) + \sin(\pimtwo - \theta) \sin\phi}{2\sin\pimtwo},  \ \hspace{5.5em} \mbox{ if}\  \theta \in [0, \pimtwo],\\
			&\frac{\cos(\theta-\pim) \cos(\phi - \frac{\pi}{2m}) + \sin(\theta - \pimtwo) \sin(\pim - \phi)}{2\sin\pimtwo},  \ \mbox{ if}\  \theta \in [\pimtwo, \pim].\\
		\end{aligned}\right.
	\end{equation*}
	For any $\phi \in [0, \pim]$, one sees $\cos\theta$ and $\sin(\pimtwo - \theta)$ are monotonically decreasing on $[0, \pimtwo]$, while $\cos(\theta-\pim)$ and $\sin(\theta - \pimtwo)$ are monotonically increasing on $[\pimtwo, \pim]$. As a result, \begin{equation*}
		G_0(\theta, \phi) \ge G_0(\pimtwo, \phi) = \frac{\cos\pimtwo \cos(\frac{\pi}{2m} - \phi)}{2\sin\pimtwo} \ge \frac{\cos^2\pimtwo}{2\sin\pimtwo} = G_0(\pimtwo, \pim)
	\end{equation*}
	for an odd $m\ge 3$.

	\item When \(\phi\in[\pim, \theta+\pim]\), we consider \(G_1(\theta, \phi-\pim)\). It is easy to check that \begin{align*}
		G_1(\theta, \phi-\pim) = \frac{\cos\theta \cos(\phi - \pim) + \sin(\frac{2\pi}{m} - \theta) \sin(\phi - \pim)}{\sin\pim} \\
		\ge G_1(\pim, \phi-\pim) = \frac{\cos(\frac{2\pi}{m} - \phi)}{\sin\pim} \ge G_1(\pim, \pim) = \frac{\cos\pim}{\sin\pim} 
	\end{align*}
	with an even \(m\). Similarly, if \(m\) is odd and \(\theta \in [0, \pimtwo]\), one has \begin{align*}
		G_1(\theta, \phi-\pim) = \frac{2\sin(\frac{2\pi}{m}-\theta) \sin(\phi-\pim) + 2\cos(\theta+\pimtwo) \cos(\phi-\pimtwo) + \sin\pim \sin(\phi - \theta)}{2\sin\pim}  \\
		\ge G_1(\pimtwo, \phi-\pim) = \frac{2\cos\pimtwo \cos(\frac{2\pi}{m} - \phi) + \sin\pim \sin(\phi - \pimtwo)}{2\sin\pim} 
		\ge G_1(\pimtwo, 0) = \frac{\cos^2\pimtwo}{2\sin\pimtwo}.
	\end{align*}
	In addition, if \(m\) is odd and \(\theta \in [\pimtwo, \pim]\), we have \begin{align*}
		G_1(\theta, \phi-\pim) = \frac{\cos(\pim-\theta) \cos(\frac{3\pi}{2m} - \phi) + \sin(\frac{3\pi}{2m} - \theta) \sin(\phi-\pim)}{2\sin\pimtwo}  \\
		\ge G_1(\theta, 0) = \frac{\cos(\pim-\theta) \cos\pimtwo}{2\sin\pimtwo} 
		\ge G_1(\pimtwo, 0) = \frac{\cos^2\pimtwo}{2\sin\pimtwo}.
	\end{align*}
\end{itemize}
In short, one has \begin{equation}\label{claim:em:2} 
	\minm{\substack{\theta\in[0, \pim] \\ \phi \in [\theta, \theta+\pim]}} G(\theta, \phi+\pitwo) = \left\{\begin{aligned}
		&\ \frac{\cos\pim}{\sin\pim},     \hspace{1.6em} \mbox{ if}\ m\ \mbox{is even},\\
		&\frac{\cos^2\pimtwo}{2\sin\pimtwo},  \quad \mbox{ if}\ m\ \mbox{is odd}.\\
	\end{aligned}\right.
\end{equation}
Substituting \eqref{claim:em:2} into \eqref{claim:em:1} gives the conclusion \eqref{eq:ortho:2}. 

Next, we consider the conclusion \eqref{eq:ortho:1} of Lemma \ref{le:ortho:l - em:p=1}. 
Without loss of generality, the condition \(\vu \perp \vv\) implies \(\theta = \phi\), i.e., \(\sumjm \ucabs{\vu^\T\va_j \va_j^\T\vv} = G(\theta, \theta+\pitwo)\). 
Therefore, from \eqref{claim:em:1}, it is easy to see that \begin{equation}\label{claim:em:3}
	\minm{0\le \theta \le \pi} G(\theta, \theta+\pitwo)  =  \minm{\theta\in[0, \pim]} G(\theta, \theta+\pitwo) = \minm{\theta\in[0, \pim]} G_0(\theta, \theta). 
\end{equation}
For an even \(m\), take \(\phi = \theta\) into \eqref{eq:gl} to obtain \begin{equation}\label{claim:em:4}
	G_0(\theta, \theta) = \frac{\cos(\pim - 2\theta)}{\sin\pim} \ge \frac{\cos\pim}{\sin\pim}
\end{equation}
with the proviso that \(\theta \in [0, \pim]\). For an odd \(m\), one has \begin{equation*}
	G_0(\theta, \theta) = \left\{\begin{aligned}
		&\frac{\cos(\frac{ \pi}{2m} - 2\theta)}{2\sin\pimtwo},   \quad \mbox{ if}\  \theta \in [0, \pimtwo],\\
		&\frac{\cos(\frac{3\pi}{2m} - 2\theta)}{2\sin\pimtwo},   \quad \mbox{ if}\  \theta \in [\pimtwo, \pim],\\
	\end{aligned}\right.
\end{equation*}
which leads to \begin{equation}\label{claim:em:5}
	G_0(\theta, \theta) \ge \frac{\cos\pimtwo}{2\sin\pimtwo}.
\end{equation}
Clearly, \eqref{claim:em:3}-\eqref{claim:em:5} constitute the conclusion \eqref{eq:ortho:1} of Lemma \ref{le:ortho:l - em:p=1}. \\
{\bf Proof of Claim \eqref{eq:gl}.}  
Recall that \begin{equation*}
	G_k(\theta, \phi) = \sum_{j=1}^{m} \ucabs{\cos(\frac{j\pi}{m} - \theta) \sin(\frac{j\pi}{m} - \phi - k\pim)}. 
\end{equation*}
Here and below, for convenience, we denote \[\tau_j^k := \cos(\frac{j\pi}{m} - \theta) \sin(\frac{j\pi}{m} - \phi - k\pim) = \frac 12 \zkh{\sin(\frac{2j\pi}{m}-\theta-\phi-\kpim)-\sin(\phi + \kpim-\theta)}.\] 
To check \(G_k(\theta, \phi)\), we divide our proof into the following two cases:

\begin{itemize}
	\item $m$ is even with $m\ge 4$. It is easy to examine \begin{align*}
		\cos(\frac{j\pi}{m}-\theta) &\ge 0 \quad \mbox{ for } 1 \le j\le \frac{m}{2}, \\
		\cos(\frac{j\pi}{m}-\theta) &\le 0 \quad \mbox{ for } \frac{m}{2} + 1 \le j \le m, \\
		\sin(\frac{j\pi}{m} - \phi - \kpim) &\ge 0 \quad \mbox{ for } k + 1 \le j\le m, \\
		\sin(\frac{j\pi}{m} - \phi - \kpim) &\le 0 \quad \mbox{ for } 1 \le j\le k .
	\end{align*}
	As a result, \begin{align*}
		G_k(\theta, \phi) =& -\sum_{j=1}^{k} \tjk + \sum_{j=k+1}^{\frac m2} \tjk - \sum_{j=\frac m2 + 1}^{m} \tjk \\
		=& (-\sum_{j=1}^{m} + 2\sum_{j=k+1}^{\frac m2}) \frac 12 \zkh{\sin(\frac{2j\pi}{m}-\theta-\phi)\cos(\kpim)-\cos(\frac{2j\pi}{m}-\theta-\phi)\sin(\kpim)} \\
		&+ k\sin(\phi-\theta+\kpim) \\
		\stackrel{\mbox{(i)}}{=}& \frac{\cos\kpim \cos(\pim - \theta - \phi)}{\sin\pim} + k\sin(\phi-\theta+\kpim). 
	\end{align*}
	where (i) comes from \eqref{eq:Lind} by the Lagrange's trigonometric identities.  
	
	\item $m$ is odd with $m\ge 3$. One can see that \begin{align*}
		\cos(\frac{j\pi}{m}-\theta) &\ge 0 \quad \mbox{ for } 1 \le j\le \frac{m-1}{2} \mbox{ if } \theta\in[0, \pimtwo], \\
		\cos(\frac{j\pi}{m}-\theta) &\le 0 \quad \mbox{ for } \frac{m+1}{2} \le j \le m \mbox{ if } \theta\in[0, \pimtwo], \\
		\cos(\frac{j\pi}{m}-\theta) &\ge 0 \quad \mbox{ for } 1 \le j\le \frac{m+1}{2} \mbox{ if } \theta\in[\pimtwo, \pim], \\
		\cos(\frac{j\pi}{m}-\theta) &\le 0 \quad \mbox{ for } \frac{m+3}{2} \le j \le m \mbox{ if } \theta\in[\pimtwo, \pim], \\
		\sin(\frac{j\pi}{m} - \phi - \kpim) &\ge 0 \quad \mbox{ for } k + 1 \le j\le m, \\
		\sin(\frac{j\pi}{m} - \phi - \kpim) &\le 0 \quad \mbox{ for } 1 \le j\le k .
	\end{align*}
	Thus, if \(\theta\in[0, \pimtwo]\), we have \begin{align*}
		G_k(\theta, \phi) =& -\sum_{j=1}^{k} \tjk + \sum_{j=k+1}^{\frac{m-1}{2}} \tjk - \sum_{j=\frac{m+1}{2}}^{m} \tjk \\
		=& \frac{\cos(\pimtwo + \kpim) \cos(\pimtwo - \theta - \phi)}{\sin\pim} + \frac{2k+1}{2}\sin(\phi-\theta+\kpim). 
	\end{align*}
	In addition, if \(\theta\in[\pimtwo, \pim]\), it holds \begin{align*}
		G_k(\theta, \phi) =& -\sum_{j=1}^{k} \tjk + \sum_{j=k+1}^{\frac{m+1}{2}} \tjk - \sum_{j=\frac{m+3}{2}}^{m} \tjk \\
		=& \frac{\cos(\pimtwo - \kpim) \cos(\frac{3\pi}{2m} - \theta - \phi)}{\sin\pim} + \frac{2k-1}{2}\sin(\phi-\theta+\kpim). 
	\end{align*}
\end{itemize}

In summary, from these two cases, we obtain \begin{equation}
	G_k(\theta, \phi) = \left\{\begin{aligned}
		&\frac{\cos\kpim \cos(\pim - \theta - \phi)}{\sin\pim} + k\sin(\phi-\theta+\kpim) \hspace{6em}  \mbox{ if}\ m\ \mbox{is even},\\ 
		&\frac{\cos(\pimtwo + \kpim) \cos(\frac{\pi}{2m} - \theta - \phi)}{\sin\pim} + \frac{2k+1}{2}\sin(\phi-\theta+\kpim) \mbox{ if}\ m\ \mbox{is odd and } \theta \in [0, \pimtwo],\\ 
		&\frac{\cos(\pimtwo - \kpim) \cos(\frac{3\pi}{2m} - \theta - \phi)}{\sin\pim} + \frac{2k-1}{2}\sin(\phi-\theta+\kpim) \mbox{ if}\ m\ \mbox{is odd and } \theta \in [\pimtwo, \pim].\\ 
	\end{aligned}\right.
\end{equation}
This completes the proof of \eqref{eq:gl}. \\
{\bf Proof of Claim \eqref{claim:em:1}.}  If \(\psi-\theta \ge \pitwo\), one has \(G(\theta, \psi) = G(\theta-k_1\pim, \psi-k_1\pim)\)
with \(k_1 = \lfloor \frac{\theta}{\pim} \rfloor\); if \(0 \le \psi-\theta < \pitwo\), we obtain that \(G(\theta, \psi) = G(\psi-\pi, \theta) = G(\psi-\pi+k_2\pim, \theta+k_2\pim)\)
with \(k_2 = \lceil \frac{\pi-\psi}{\pim} \rceil\). This implies 
\[\minm{0\le \theta, \psi \le \pi} G(\theta, \psi) = \minm{\substack{\theta\in[0, \pim] \\ \psi \in [\frac{\pi}{2} + \theta, \pi]}} G(\theta, \psi) = \minm{\substack{\theta\in[0, \pim] \\ \phi \in [\theta, \pitwo]}} G(\theta, \phi+\pitwo) . \]
Thus, we only need to verify   \begin{equation}\label{eq:gl+1}
	\minm{\substack{\theta\in[0, \pim] \\ \phi \in [\theta, \pitwo]}} G(\theta, \phi+\pitwo) = \minm{\substack{\theta\in[0, \pim] \\ \phi \in [\theta, \theta+\pim]}} G(\theta, \phi+\pitwo). 
\end{equation}
In fact, it suffices to prove that \begin{align}
	\label{eq:gl+1:1} G_{k+1}(\theta, \phi) \ge G_{k}(\theta, \phi) \quad \mbox{for } 0 \le \phi \le \theta \le \pim \mbox{ and } k \ge 1, \\
	\label{eq:gl+1:0} G_{k+1}(\theta, \phi) \ge G_{k}(\theta, \phi) \quad \mbox{for } 0 \le \theta \le \phi \le \pim \mbox{ and } k \ge 0. 
\end{align}
To determine \eqref{eq:gl+1:1}, we consider three cases:

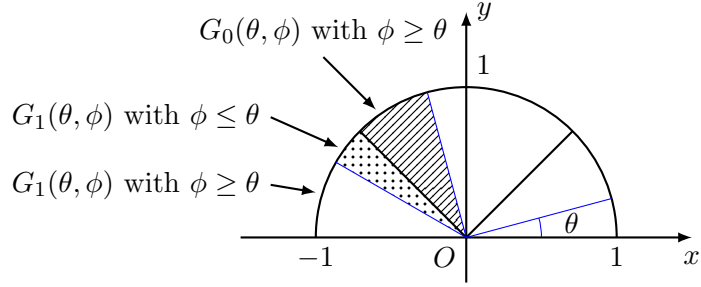
\begin{figure}[t]
	\setlength\tabcolsep{1pt}
	\centering
	\begin{tikzpicture}[>=latex,scale=2]
		\draw[->, thick] (-1.5,0) -- (1.5,0) node[below, yshift=-0.05cm] {$x$};
		\draw[->, thick] (0,-0.3) -- (0,1.5) node[right] {$y$};
		\node at (0,0) [below left] {$O$};

		\draw[thick] (-1,0) arc (180:0:1);
		
		\foreach \x in {-1,1}
		\draw (\x,0) -- (\x,-0) node[below] {$\x$};
		\foreach \y in {1}
		\node[right, yshift=0.3cm] at (0,\y) {$\y$}; 
		
		
		\begin{scope}
			\path[pattern=north east lines] (0,0) -- (135:1) arc (135:105:1) -- cycle;
			
			\path[pattern=dots] (0,0) -- (150:1) arc (150:135:1) -- cycle;
		\end{scope}
		
		\draw[thick] (45:1)--(0,0)--(135:1);
		\draw[blue] (105:1)--(0,0)--(15:1);
		\draw[blue] (150:1)--(0,0);
		
		\draw[<-, thick] (-0.6,0.85) -- (-0.95,1.2) node[above] {$G_0(\theta, \phi)$ with $\phi \ge \theta$};
		\draw[<-, thick] (-0.85,0.6) -- (-1.3,0.8) node[left] {$G_1(\theta, \phi)$ with $\phi \le \theta$};
		\draw[<-, thick] (-1,0.3) -- (-1.3,0.35) node[left] {$G_1(\theta, \phi)$ with $\phi \ge \theta$};
		
		\draw[blue] (0.5,0) arc (0:15:0.5);
		\node at (0.7,0) [above, yshift=-0.08cm] {$\theta$};
	\end{tikzpicture}
	\caption{A schematic diagram when \( m = 4 \) and \( \theta = \frac{\pi}{12} \).}
	\label{S1}
\end{figure}

{\bf Case 1:} $m$ is even with $m\ge 4$. Recalling \eqref{eq:gl}, it holds \begin{align*}
	G_{k+1}(\theta, \phi) - G_{k}(\theta, \phi) =& 2k\cos(\phi-\theta+\kpim+\pimtwo)\sin\pimtwo + \sin(\phi - \theta + \frac{k+1}{m}\pi) \\
	&-\frac{2\sin(\kpim+\pimtwo)\sin\pimtwo\cos(\pim-\theta-\phi)}{\sin\pim} \\
	\stackrel{\mbox{(i)}}{\ge}& 2k\cos(\kpim+\frac{\pi}{2m})\sin\pimtwo + \sin\kpim - \frac{2\sin(\kpim+\pimtwo)\sin\pimtwo}{\sin\pim}\\
	=& 2k\cos(\kpim+\frac{\pi}{2m})\sin\pimtwo - \frac{2\cos\kpim\sin^2\pimtwo}{\sin\pim}\\
	=& \frac{2\sin^2\pimtwo}{\sin\pim}\zkh{2k \cos(\kpim+\pimtwo)\cos\pimtwo-\cos\kpim} \stackrel{\mbox{(ii)}}{\ge} 0.
\end{align*}
Here, (i) follows from the facts $\cos(\phi-\theta+\kpim+\pimtwo)\ge \cos(\kpim+\frac{\pi}{2m})$, $\sin(\phi - \theta + \frac{l+1}{m}\pi)\ge \sin\kpim$ and $\cos(\pim-\theta-\phi)\le 1$ with \(k \le \hat{k}-1 = \frac m2 -2\). And (ii) holds since 
\[2k \cos(\kpim+\pimtwo)\cos\pimtwo \ge 2\cos(\kpim+\pimtwo)\cos\pimtwo = \cos\kpim + \cos(\kpim+\pim) \ge \cos\kpim\]
with the proviso that \(1 \le k \le \frac m2 -2\).

{\bf Case 2:} $m$ is odd with \(\theta\in[0, \pimtwo]\) and $\phi\in[0, \theta]$. \eqref{eq:khat} gives \(m \ge 5\) and \(1 \le k \le \frac{m-3}{2}\). Then \eqref{eq:gl} asserts that \begin{align*}
	G_{k+1}(\theta, \phi) - G_{k}&(\theta, \phi) = (2k+1)\cos(\phi-\theta+\kpim+\pimtwo)\sin\pimtwo + \sin(\phi - \theta + \frac{k+1}{m}\pi) \\
	&-\frac{2\sin(\kpim+\pim)\sin\pimtwo\cos(\pimtwo-\theta-\phi)}{\sin\pim} \\
	\ge& (2k+1)\cos(\kpim+\pimtwo)\sin\pimtwo + \sin(\kpim+\pimtwo) - \frac{2\sin(\kpim+\pim)\sin\pimtwo}{\sin\pim}\\
	=& (2k+1)\cos(\kpim+\pimtwo)\sin\pimtwo - \frac{2\cos(\kpim+\pimtwo)\sin^2\pimtwo}{\sin\pim}\\
	=& \frac{2\sin^2\pimtwo\cos(\kpim+\pimtwo)}{\sin\pim}\zkh{ (2k+1) \cos\pimtwo - 1} \stackrel{\mbox{(i)}}{\ge} 0 
\end{align*}
where (i) holds because of \((2k+1) \cos\pimtwo \ge 3 \cos\frac{\pi}{10}  >  1.\)

{\bf Case 3:} $m$ is odd with \(\theta\in[\pimtwo, \pim]\) and $\phi\in[0, \theta]$. We first prove \(G_{2}(\theta, \phi) \ge G_{1}(\theta, \phi)\) with \(m \ge 5\) due to \(\hat{k} \ge 2\). One has 
\begin{align*}
	G_{2}(\theta, \phi&) - G_{1}(\theta, \phi) = 3\cos(\phi-\theta+\pimtwothree)\sin\pimtwo + \sin(\phi - \theta + \pim) - 2\sin\pimtwo\cos(\pimtwothree-\theta-\phi) \\
	=& -4\sin\pimtwo\sin(\pimtwothree-\theta)\sin\phi + \cos\pimtwo\sin(\pimtwothree-\theta)\cos\phi + \cos\pimtwo\cos(\pimtwothree-\theta) \\
	=& \cos\pimtwo\cos(\pimtwothree-\theta)\cos\phi \zkh{-4\tan\pimtwo\tan(\pimtwothree-\theta)\tan\phi + \tan(\pimtwothree-\theta) + \tan\phi}\\
	=& \frac{\cos\pimtwo\cos(\pimtwothree-\theta)\cos\phi}{4\tan\pimtwo} \zkh{1 - \xkh{1-4\tan\pimtwo\tan(\pimtwothree-\theta)}\xkh{1-4\tan\pimtwo\tan\phi} } \stackrel{\mbox{(i)}}{\ge} 0.
\end{align*}
Here, (i) holds since 
\[\tan\pimtwo\tan(\pimtwothree-\theta),\ \tan\pimtwo\tan\phi \le \tan\pimtwo\tan\pim \le \tan\frac{\pi}{10}\tan\frac{\pi}{5} = \sqrt{5}-2 < \frac 14\]
for \(m \ge 5\). Then consider \(G_{k+1}(\theta, \phi) - G_{k}(\theta, \phi)\) with \(k \ge 2\) and \(\hat{k} \ge 3\), i.e. \(m \ge 7\). It is easy to see that \begin{align*}
	G_{k+1}(\theta, \phi) - G_{k}(&\theta, \phi) \ge (2k+1)\cos(\kpim+\pimtwo)\sin\pimtwo + \sin(\kpim-\pim) - \frac{2\sin\kpim \sin\pimtwo}{\sin\pim}\\
	=& (2k+1)\cos(\kpim+\pimtwo)\sin\pimtwo + \frac{2\sin\pimtwo}{\sin\pim} \zkh{\sin(\kpim-\pim)\cos\pimtwo - \sin\kpim}\\
	\ge& (2k+1)\cos(\kpim+\pimtwo)\sin\pimtwo + \frac{2\sin\pimtwo}{\sin\pim} \zkh{\sin(\kpim-\pim)\cos\pim - \sin\kpim}\\
	=& \sin\pimtwo \Big[\underbrace{(2k+1) \cos(\kpim+\pimtwo)-2\cos(\kpim-\pim)}_{:=I_1}\Big] \ge 0. 
\end{align*}
Here,  we consider \(I_1\) for \(3 \le k \le \frac{m-3}{2}\) and \(k = 2\), respectively.
\begin{itemize}
	\item If \(3 \le k \le \frac{m-3}{2}\), we have \begin{align*}
		I_1 =& (2k+1) \cos(\kpim+\pimtwo) - 2\cos(\kpim+\pimtwo)\cos\pimtwothree - 2\sin(\kpim+\pimtwo)\sin\pimtwothree \\
		\stackrel{\mbox{(i)}}{\ge}& (2+4\cos\pimtwothree) \cos(\kpim+\pimtwo) - 2\cos(\kpim+\pimtwo)\cos\pimtwothree - 2\sin(\kpim+\pimtwo)\sin\pimtwothree \\
		=& 2 \cos(\kpim+\pimtwo) + 2 \cos(\kpim+\frac{2\pi}{m}) = 4 \cos(\kpim+\frac{5\pi}{4m}) \cos \pimfourthree \stackrel{\mbox{(ii)}}{\ge} 0
	\end{align*}
	where (i) uses \(2k+1 >  2 + 4\cos\pimtwothree\) for \(k \ge 3\), and (ii) holds since \(3 \le k \le \frac{m-3}{2}\). 
	\item If \(k = 2\), take \(x = \pimtwo\) and see that \[I_1 = 5\cos\frac{5\pi}{2m}-2\cos\pim = 5\cos(5x) - 2\cos(2x).\]
	We denote \(f(x) := 5\cos(5x) - 2\cos(2x)\) with \(x \in (0, \frac{\pi}{14}]\) and obtain  
	\[f'(x) = -25\sin(5x) + 4\sin(2x) < -25\sin(2x) + 4\sin(2x) < 0.\] 
	As a result, one has \(f(x)>0\) for \(x \in (0, \frac{\pi}{14}]\) since \(f(0) = 3\) and 
	\[f(\frac{\pi}{14}) = 5\cos\frac{5\pi}{14}-2\cos\frac{\pi}{7} = 3\cos\frac{5\pi}{14}-4\sin\frac{\pi}{4}\sin\frac{3\pi}{28} > 3\sin\frac{4\pi}{28}-3\sin\frac{3\pi}{28} > 0.\]
\end{itemize}

The above analysis establishes \eqref{eq:gl+1:1}. 
Inequality \eqref{eq:gl+1:0} is similar; for the sake of simplicity, we omit its proof here. 
This completes the proof. 


\bibliographystyle{plain}

\end{document}